\documentclass[onecolumn,prx, superscriptaddress,longbibliography,margin=1mm 11pt]{revtex4-2}

\usepackage{graphicx}
\usepackage{bm}
\usepackage{amsmath}
\usepackage{environ}
\usepackage{appendix}
\usepackage{physics}
\usepackage{array}
\usepackage{csquotes}
\usepackage{amsthm}
\usepackage{enumitem}
\usepackage{qcircuit}
\usepackage{subcaption}
\graphicspath {{./assets/}}
\usepackage{nicefrac}

\usepackage{mathtools}

\newcommand{\myeq}[1]{\stackrel{\mathclap{\footnotesize\mbox{#1}}}{=}}

\usepackage[nopar]{lipsum}

\usepackage{amsthm}
\usepackage{amssymb}
\usepackage{amsfonts}
\usepackage{fancyhdr}
\usepackage{slashed}
\usepackage{bbm}
\usepackage{latexsym,epsfig,bbm}
\usepackage[colorlinks=true,urlcolor=blue,citecolor=blue]{hyperref}

\usepackage{color}
\definecolor{blue}{rgb}{0,0.2,1}

\definecolor{red}{rgb}{0.9,0,0}

\newcommand{\vect}[1]{\boldsymbol{#1}}
\newtheorem{theorem}{Theorem}
\newtheorem{lemma}[theorem]{Lemma}
\newtheorem{definition}[theorem]{Definition}

\newtheorem{example}[theorem]{Example}
\newtheorem{remark}[theorem]{Remark}

\usepackage{etoolbox}
\begin{document}

\title{Quantum simulation for time-dependent Hamiltonians -- with applications to non-autonomous ordinary and partial differential equations}

\date{\today}
\author{Yu Cao}
\affiliation{School of Mathematical Sciences, Shanghai Jiao Tong University, Shanghai, 200240, China}
\affiliation{Institute of Natural Sciences, Shanghai Jiao Tong University, Shanghai 200240, China}
\affiliation{Ministry of Education Key Laboratory in Scientific and Engineering Computing, Shanghai Jiao Tong University, Shanghai 200240, China}

\author{Shi Jin}
\affiliation{School of Mathematical Sciences, Shanghai Jiao Tong University, Shanghai, 200240, China}
\affiliation{Institute of Natural Sciences, Shanghai Jiao Tong University, Shanghai 200240, China}
\affiliation{Ministry of Education Key Laboratory in Scientific and Engineering Computing, Shanghai Jiao Tong University, Shanghai 200240, China}

\affiliation{Shanghai Artificial Intelligence Laboratory, Shanghai, China}

\author{Nana Liu}
\email{nana.liu@quantumlah.org}
\affiliation{School of Mathematical Sciences, Shanghai Jiao Tong University, Shanghai, 200240, China}
\affiliation{Institute of Natural Sciences, Shanghai Jiao Tong University, Shanghai 200240, China}
\affiliation{Ministry of Education Key Laboratory in Scientific and Engineering Computing, Shanghai Jiao Tong University, Shanghai 200240, China}
\affiliation{Shanghai Artificial Intelligence Laboratory, Shanghai, China}
\affiliation{University of Michigan-Shanghai Jiao Tong University Joint Institute, Shanghai 200240, China.}

\begin{abstract} 
Non-autonomous dynamical systems appear in a very wide range of interesting applications, both in classical and quantum dynamics, where in the latter case it  corresponds to having a time-dependent Hamiltonian. However, the quantum simulation of these systems often needs to appeal to rather complicated procedures involving the Dyson series, considerations of time-ordering, requirement of time steps to be discrete and/or requiring multiple measurements and postselection. These procedures are generally much more complicated than the quantum simulation of time-independent Hamiltonians. Here we propose an alternative formalism that turns any non-autonomous unitary dynamical system into an autonomous unitary system, i.e., quantum system with a time-{\it independent} Hamiltonian, in one higher dimension, \textit{while keeping time continuous}. This makes the simulation with time-dependent Hamiltonians not much more difficult than that of time-independent Hamiltonians, and can also be framed in terms of an analogue quantum system evolving continuously in time. We show how our new quantum protocol for time-dependent Hamiltonians can be performed in a resource-efficient way and without measurements, and can be made possible on either continuous-variable, qubit or hybrid systems. Combined with a technique called Schr\"odingerisation, this dilation technique can be applied to the quantum simulation of any linear ODEs and PDEs, and nonlinear ODEs and certain nonlinear PDEs, with time-dependent coefficients. 
\end{abstract}
\maketitle 

\section{Introduction}
Non-autonomous dynamical systems \cite{kloeden2011nonautonomous, kiguradze2012asymptotic, gunzburger2007reduced} appear in broad areas of applications, in both classical \cite{clemson2014discerning} and quantum physical systems \cite{poulin2011quantum}, in adiabatic quantum simulation \cite{albash2018adiabatic}, in biological systems \cite{thieme2000uniform, coutinhoa2006threshold}, and in finance such as the Black-Scholes equations with time-dependent coefficients -- for example the  risk-free interest rate \cite{lo2003simple}. Numerical simulation of such systems are challenging \cite{gunzburger2007reduced}, especially when the time-dependent operator in the system do not commute at different times. In this case, the evolution operator needs time-ordering which is usually interpreted by Dyson's theory \cite{reed1975methods}, in the case of unitary dynamics. Quantum simulation of such systems -- with time-dependent Hamiltonians -- also become more difficult. For instance, when using Trotter splitting, each step in the split involves a different Hamiltonian,  corresponding to {\it different quantum circuits at different times.}\\

In this paper we introduce a strategy based on extending non-autonomous dynamical systems, including both ordinary differential equations (ODEs) and partial differential equations (PDEs), to new systems of autonomous PDEs in one higher dimension. This new dimension, which can be treated like a spatial dimension, gives rise to an extra linear convection term and acts like a clock parameter. While previous classical methods \cite{guckenheimer2013nonlinear, wiggins1991introduction} have also used such a clock parameter, they turn non-autonomous systems of linear ODEs into a system of nonlinear ODEs in one higher dimension. In our approach, the autonomous PDEs in one higher dimension {\it remains linear}. Such linearity not only reduces resources but is also much more suitable for quantum simulation. \\

This idea can be naturally used for unitary dynamics, allowing us to avoid using time ordering and Dyson's series -- which is the foundation of many of the existing quantum algorithms for time-dependent Hamiltonians (e.g ., \cite{poulin2011quantum,wiebe2010higher, low2018hamiltonian,berry2020time,chen2020quantum,AnFangLin, an2022time}) in constructing similarly efficient quantum algorithms. Our formulation has several key benefits: 1) We do not require discrete time-steps at the outset, which is necessary, for instance, for a formalism based on Dyson's series. Our formulation instead transforms the original problem into unitary evolution with a time-independent Hamiltonian in \textit{continuous time} using only one additional mode, thus making it also a candidate for analogue quantum simulation. 2) Unlike many previous methods, we do not require any measurements or postselection. 3)  The initial ancilla states can also be a fully mixed state, which can further simplify implementation. 4) Furthermore, unlike some previous algorithms, we also do not require extra access to operations dealing with time-ordering or the requirement of implementing a different unitary at a different time-step. \\

We can easily extend this formalism also to non-unitary dynamics, while still easily maintaining continuity in the time parameter, as well as all other parameters. For this we will use the  Schr\"odingerisation technique \cite{2023analog,jin2212quantum, jin2022quantum}, which can turn any non-unitary dynamics into unitary dynamics while preserving continuity in time. Schr\"odingerisation can thus be directly integrated with our new method. This means that no discretisation of either time or spatial parameters at any stage is necessary for the formalism to apply. \\

Our protocol finds applications not only in linear ODEs and PDES with time-dependent coefficients, but also some nonlinear PDEs including scalar hyperbolic balance laws, and Hamilton-Jacobi equations. These systems can be transferred to linear transport equations via the level set techniques \cite{JinOsher, 2022nonlinearpde}, and then the aforementioned strategy applies. \\

Before moving on, we also mention here a few related, but in-equivalent, works in the mathematical physics community and in the foundations of quantum physics. To the best of our knowledge, Howland firstly proposed an augmented Hamiltonian, also known as \enquote{Howland’s clock Hamiltonian} to study scattering theories \cite{howland_stationary_1974}. This formalism was recently adopted by e.g., \citet{burgarth_control_2023} to construct a  dilation method for open quantum systems. However, these tools were used more as a mathematical method for the study of quantum systems rather than an explicit proposal of a quantum protocol to allow the quantum simulation of non-autonomous systems, which we demonstrate explicitly in this paper along with various applications. We emphasise that having the same Hamiltonian does not mean that one has the same protocol. In the foundations of quantum physics community, a related set of works relies on the Page-Wooters mechanism \cite{page1983evolution}, which looks at turning time-dependent quantum dynamics into time-independent quantum dynamics, rather than just autonomous quantum mechanics. It can be used for instance in finding a quantum description for time \cite{giovannetti2015quantum}. We briefly summarise the mechanism in Appendix~\ref{app:PW}. Although it is clear that it is distinct from our proposal, it is interesting to speculate on how our protocol could inform future studies in this foundational area. \\

Although our initial formulation is chiefly based on continuous-variables, this is for the purpose of conveying the central ideas in as clean a manner as possible. Our approach can be easily adopted to either continuous-variable quantum systems (qumodes), qubits or hybrid qubit-qumode formulations, which we describe later in the paper. \\

In Section~\ref{sec:nonautotransformation} we will formulate our main idea of turning a non-autonomous unitary system into a linear autonomous PDE system with unitary dynamics, so this applies also to classical simulation. In Section~\ref{sec:quantumsimintro} we describe how general linear non-autonomous systems can be mapped onto quantum simulation with a time-independent Hamiltonian operating only on one extra mode. We look at non-autonomous unitary dynamics and non-unitary dynamics separately, where in the latter case we apply Schr\"odingerisation. In this section, we also describe the qubit formulation and hybrid formulations. See Figures~\ref{fig:quantumdynamics} and~\ref{fig:perfect} for a summary of the protocols. In Section~\ref{sec:applications} we describe applications to quantum systems and more general linear and nonlinear ODEs and PDEs. In Section~\ref{sec:numerical},  we demonstrate the effectiveness of our quantum algorithms via numerical examples, in particular, emulating a 1D time-dependent Fokker-Planck equation with 20 qubits in total. Finally, we include a brief discussion in Section~\ref{section::discussion}.

\subsection*{Notation}
Here we summarise our notation below. A vector $\vect{u}(t)$ is time-dependent and in the ket-notation $|u(t)\rangle \equiv \vect{u}(t)/\|\vect{u}(t)\|$ denotes its normalised counterpart, where $\| . \|$ is the $l_2$-norm. When $\vect{u}(t)$ obeys unitary evolution, then $\|\vect{u}(t)\|=\|\vect{u}(0)\|$, where we can assume $\|\vect{u}(0)\|=1$ unless otherwise stated, so $|u(t)\rangle=\vect{u}(t)$. According to the context, $\vect{u}(t)$ could act on either finite-dimensional or infinite-dimensional Hilbert spaces. \\

Quantum information processing protocols are often introduced in terms of qubits, which are the quantum counterparts to classical bit-strings $\{0,1\}^{\otimes m}$. A qubit uses the eigenbasis $\{|0\rangle, |1\rangle\}$ for the computational basis. An $J$-qubit system is a tensor product of $J$ qubits which lives in $2^J$-dimensional Hilbert space, so $|u(t)\rangle \propto \sum_{j=1}^J \vect{u}_j(t)|j\rangle$, where $\vect{u}_j(t)$ is the $j^{\text{th}}$-component of $\vect{u}(t)$. A continuous-variable (CV) quantum state, or {\it `qumode'}, on the other hand, spans an infinite-dimensional Hilbert space. A qumode is the quantum analogue of a continuous classical degree of freedom, like position, momentum or energy before being quantised. A qumode is equipped with observables with a continuous spectrum, such as the position $\hat{x}$ and momentum $\hat{p}$ observables of a quantum particle. Its eigenbasis can be chosen to be for instance $\{|x\rangle\}_{x \in \mathbb{R}}$, which are the eigenstates of $\hat{x}$. A system of $D$-qumodes is a tensor product of $D$ qumodes. The qubit and qumode language can be easily interchanged as well, in the limit where we are restricted to a truncation of the infinite-dimensional Hilbert space. In this paper, although most of the language is in the qumode of CV language, this is done for simplicity only, and the results are easily translated into that of qubits.\\

In this paper, we use $[\hat{A}, \hat{B}]=\hat{A}\hat{B}-\hat{B}\hat{A}$ to denote commutation relations and $\{\hat{A}, \hat{B}\}=\hat{A}\hat{B}+\hat{B}\hat{A}$ to denote anticommutation relations. 

\section{Transforming a non-autonomous system to an autonomous system with unitary dynamics} \label{sec:nonautotransformation}

We consider a general linear non-autonomous dynamical system $\vect{u}(t)$ on a Hilbert space $\mathcal{H}$, 
\begin{align}\label{ODE}
&\frac{d\vect{u}(t)}{dt}= -i\vect{A}(t) \vect{u}(t) \qquad \vect{u}(0)=\vect{u}_0. 
\end{align}
Here $\vect{u}(t)$ can be either a finite-dimensional vector $\vect{u}(t)=\sum_{j=1}^J u_j(t)|j\rangle$ (for $J$-dimensional ODEs), or an infinite-dimensional vector $\vect{u}(t)=\int_{\mathbb{R}^D} u(t, x_1, \cdots, x_D)|x_1, \cdots, x_D\rangle dx_1 \cdots dx_D$, where $\vect{A}(t)$ are  differential operators with (non-linear) time-dependent coefficients.
Our aim is to transform this into an autonomous system with unitary dynamics.

Our first observation is that any system in Eq.~\eqref{ODE} can be easily transformed into a linear non-autonomous system with unitary dynamics
\begin{align} \label{eq:timeschr}
\begin{aligned}
    & \frac{d\vect{y}(t)}{dt}=-i\vect{H}(t)\vect{u}(t), \qquad \vect{H}(t)=\vect{H}^{\dagger}(t) , \\
    & \vect{y}(0)=\vect{y}_0.
    \end{aligned}
\end{align}
This is clear in the case when $\vect{A}$ is Hermitian: $\vect{A}(t)=\vect{A}^{\dagger}(t)=\vect{H}(t)$, which describes a quantum system $\vect{u}(t)=\vect{y}(t)$ evolving unitarily under a time-dependent Hamiltonian. The second case is the more general $\vect{A}(t) \neq \vect{A}^{\dagger}(t)$ scenario. In this case we can use a method called Schr\"odingerisation \cite{2023analog, jin2022quantum, jin2212quantum}, where $\vect{H}(t)$ can be easily constructed from $\vect{A}(t)$ and acts on a Hilbert space with only one extra mode. Then there is another simple step to recover $\vect{u}(t)$ from $\vect{y}(t)$. We will review this Schr\"odingerisation method in Section~\ref{sec:quantumsimintro}. \\

We consider the time-dependent Hamiltonian system in Eq~\eqref{eq:timeschr}. 
The solution to Eq.~\eqref{eq:timeschr} can be written as 
\begin{equation}\label{ODE-solution}
\vect{y}(t)=\mathcal{U}_{t, 0}\vect{y}_0, \qquad  
\end{equation}
where 
\begin{align*} 
 \mathcal{U}_{t, s}=   \mathcal{T}e^{-i\int_s^{t}\vect{H}(\tau) d\tau}=\lim_{N \rightarrow \infty} e^{-i\vect{H}(t_N)\Delta t}\cdots e^{-i\vect{H}(t_1)\Delta t}
=I+\sum_{n=1}^\infty (-i)^n\frac{1}{n!}\int_s^t dt_1 \cdots \int_s^t dt_n \mathcal{T}\vect{H}(t_1)\cdots \vect{H}(t_n),
\end{align*}
and $\mathcal{T}$ is the chronological time-ordering operator. 
Note that in the case where $[\vect{H}(t), \vect{H}(t')]=0$, then no time-ordering is required since $\vect{H}(t)$ commutes at different times. It is important to note that when $[\vect{H}(t), \vect{H}(t')] \neq 0$, $\exp(-i\int_0^t \vect{H}(\tau)d\tau)$ is \textit{not} a solution to Eq.~\eqref{eq:timeschr} and the time-ordering operation is necessary to define an ordering of the operations $\exp(-i\vect{H}(t_i) \Delta t)$ where $t_N>t_{N-1}>\cdots>t_1$ with $t=N\Delta t$.
\\

The unitary operator $\mathcal{U}_{t, s}$ satisfies the following properties \cite{reed1975methods}:

\begin{lemma}
For any $t\in \mathbb{R}$, 
\begin{align}
& \frac{d \mathcal{U}_{t,s}}{dt}=-i\vect{H}(t)\mathcal{U}_{t,s},
\qquad \frac{d \mathcal{U}_{s,t}}{dt}=i\,\mathcal{U}_{s,t}\vect{H}(t),\label{eq:vequation} \qquad  \mathcal{U}_{t,t}=\mathbf{1},
\end{align}
and for any $s'\in \mathbb{R}$, 
\begin{align} \label{eq:vequation2}
& \mathcal{U}_{t, s}=\mathcal{U}_{t, s'}\mathcal{U}_{s', s},  \qquad \mathcal{U}^{\dagger}_{t,s}=\mathcal{U}_{s,t}. 
\end{align}
\end{lemma}
However, solving the problem in Eq.~\eqref{eq:timeschr} through this time-ordering is complicated. The difficulty for this non-autonomous system is that one needs to evolve according to the time-dependent operator $\mathcal{T}\exp(-i\int_{0}^t\vect{H}(\tau)d\tau)$ in a chronological way. For instance, in quantum simulation, not only are time-ordered oracles necessary, but also at each time-interval, a {\it different} Hamiltonian $\vect{H}(t_i)$, corresponding to a different gate, is required. \\

In this section we instead propose a reformulation of the initial value problem in  Eq.~\eqref{eq:timeschr} to arrive at a dynamics that will only evolve according to a {\it time-independent} linear operator. In the case of Hamiltonian simulation for instance, this gives rise to a corresponding time-independent Hamiltonian. The idea is to introduce a new time variable $s$ so that the problem becomes a new PDE system defined in one higher dimension, now with {\it time-independent} coefficients! 

\begin{theorem} \label{thm:one}
For the non-autonomous system in Eq.\eqref{eq:timeschr}, we introduce the following initial-value problem of an  {\it autonomous} PDE
\begin{align}\label{w-eqn}
\begin{aligned}
&\frac{\partial \vect{w}}{\partial t}+ \frac{\partial \vect{w}}{\partial s} = -i\vect{H}(s) \vect{w}
\\
& \vect{w}(0,s)=G(s)\vect{y}_0, \qquad s \in \mathbb{R}.
\end{aligned}
\end{align}
The analytical solution to this problem is
\begin{equation}\label{w-solutiongeneral}
\vect{w}(t,s)=G(s-t)\mathcal{U}_{s, s-t}\vect{y}_0, \qquad \mathcal{U}_{s,s-t}=\mathcal{T}e^{-i\int_{s-t}^s \vect{H}(\tau)d \tau}=\mathcal{T}e^{-i \int_{0}^t \vect{H}(s-t+\tau)d\tau}.
\end{equation}
 When $G(s)=\delta(s)$, one can easily recover $\vect{y}(t)$ in Eq.~\eqref{eq:timeschr} from $\vect{w}(t,s)$ using 
\begin{align}
\label{eqn::yt}
 \vect{y}(t)=\int^{\infty}_{-\infty} \vect{w}(t,s)\, ds.   
\end{align}
Alternatively, when $G(s)=1$, $\vect{y}(t)$ can be recovered with $\vect{w}(t,s=t)=\vect{y}(t)$. 

\end{theorem}

\begin{proof}[Proof of Theorem~\ref{thm:one}]

We will now prove that $\vect{w}(t,s)=G(\sigma)\mathcal{U}_{s,\sigma}\vect{y}_0$ solves Eq.~\eqref{w-eqn} using $\sigma=s-t$. The LHS of Eq.~\eqref{w-eqn} can be written as
\begin{align*}
   \frac{\partial \vect{w}}{\partial t}+\frac{\partial \vect{w}}{\partial s} &=\left(\frac{\partial G(\sigma)}{\partial t}+\frac{\partial G(\sigma)}{\partial s}\right)\mathcal{U}_{s,\sigma}\vect{y}_0+G(\sigma)\left(\frac{\partial \mathcal{U}_{s,\sigma}}{\partial t}+ \frac{\partial \mathcal{U}_{s,\sigma}}{\partial s}\right)\vect{y}_0 \\
   &=G(\sigma)\left(\mathcal{U}_{s,0}\frac{\partial \mathcal{U}_{0, \sigma}}{\partial t}+\frac{\partial \mathcal{U}_{s,0}}{\partial s}\mathcal{U}_{0, \sigma}+\mathcal{U}_{s,0}\frac{\partial \mathcal{U}_{0,\sigma}}{\partial s}\right)\vect{y}_0\\
   &=G(\sigma)\frac{\partial \mathcal{U}_{s,0}}{\partial s}\mathcal{U}_{0, \sigma}\vect{y}_0+G(\sigma)\mathcal{U}_{s,0}\left(\frac{\partial \mathcal{U}_{0,\sigma}}{\partial t}+\frac{\partial \mathcal{U}_{0,\sigma}}{\partial s}\right)\vect{y}_0 \\
   &=-i\vect{H}(s) \vect{w}+G(\sigma)\mathcal{U}_{s,0}\left(\frac{\partial \mathcal{U}_{0,\sigma}}{\partial t}+\frac{\partial \mathcal{U}_{0,\sigma}}{\partial s}\right)\vect{y}_0\\
 &=-i\vect{H}(s) \vect{w}.  
\end{align*}
Here in the first line the first term in brackets is zero, since letting $\sigma=s-t$, clearly $\partial G(\sigma)/\partial s=-\partial G(\sigma)/\partial t$. In the second term we used the expansion $\mathcal{U}_{s,\sigma}=\mathcal{U}_{s,0}\mathcal{U}_{0,\sigma}$. In the third line we used the definition for $\mathcal{U}_{s,0}$ which obeys Eq.~\eqref{eq:vequation} and $G(\sigma)\mathcal{U}_{s,0}\mathcal{U}_{0,\sigma}\vect{y}_0=\vect{w}(t,s)$. The second term in the third line goes to zero since $
    \partial \mathcal{U}_{0, \sigma}/{\partial s}=-\partial \mathcal{U}_{0, \sigma}/{\partial t}$.

When $G(s)=\delta(s)$, Theorem~\ref{thm:one} easily follows by integrating $\int ds\, \vect{w}(t,s)=\int ds \,\delta(s-t) \mathcal{U}_{s,s-t} \vect{y}_0=\mathcal{U}_{t,0}\vect{y}_0=\vect{y}(t)$. \\
When $G(s)=1$, $\vect{w}(t, s)=\mathcal{U}_{s, s-t}\vect{y}_0$, so $\vect{w}(t,s=t)=\mathcal{U}_{t, 0}\vect{y}_0=\vect{y}(t)$.
\end{proof}

\begin{remark}
    While Theorem~\ref{thm:one} holds for Eq.~\eqref{eq:timeschr} with unitary dynamics,  the same result can also hold for \textit{non-unitary dynamics} or $\vect{H}(t) \neq \vect{H}^{\dagger}(t)$ in the case when $[\vect{H}(t), \vect{H}(t')]=0$ for all $t, t'$. In this case no time-ordering is required in the solution, and $\mathcal{U}_{s, s-t}=\exp(-i\int_{s-t}^s \vect{H}(\tau) d\tau)=\mathcal{U}_{s, 0}\mathcal{U}_{0, s-t}$. Then it is clear that the same proof for Theorem~\ref{thm:one} still holds. However, in the more general case $[\vect{H}(t), \vect{H}(t')] \neq 0$ when $\vect{H}(t) \neq \vect{H}^{\dagger}(t)$, we can no longer write $\mathcal{U}_{s,s-t}=\mathcal{U}_{s,0}\mathcal{U}_{0,s-t}$ and the time-ordered Dyson series itself may not even converge. We leave further investigation on this topic to future work.
    \end{remark}

 \begin{remark} A conventional way to transform a non-autonomous system to an autonomous ones is to add a new variable representing time \cite{guckenheimer2013nonlinear,wiggins1991introduction}
\begin{align}
\begin{aligned}
 &\frac{d \vect{y}}{dt}=-i \vect{H}(\tau)\vect{y}(t)\\
 & \frac{d\tau}{dt}=1\\
 & \vect{y}(0)=\vect{y}_0, \quad \tau(0)=0.
\end{aligned}
\label{ODE-oldidea}
\end{align}
Note that even if the original system is linear, this new autonomous system becomes {\it nonlinear}. On the other hand, 
our new formulated  PDEs \eqref{w-eqn} {\it remains
linear}! Maintaining linearity is crucial for the application to the quantum simulation of non-autonomous PDEs, including applications to time-dependent Hamiltonian simulation. This method in Eq.~\eqref{ODE-oldidea} is actually more akin to the simulation method described in Appendix C.
\end{remark}

\begin{remark}
If the original system is already independent of time, then 
\begin{align*}
\vect{w}(t,s) = G(s-t) e^{-i \vect{h} t} \vect{y_0}.
\end{align*}
One can easily show that after tracing out the degree of freedom of $s$ (whether classically as in Eq.~\eqref{eqn::yt} or quantum mechanically as in Eq.~\eqref{eqn::gamma}), we can readily verify that e.g., $\vect{y}(t) = e^{-i \vect{h} t}\vect{y_0}$ for any distribution $G$. Such a consistency conclusion implies that adding such an extra degree of freedom won't affect a given time-independent Hamiltonian problem.
\end{remark}

For classical simulation, $G(s)=1$ appears like a good choice to recover $\vect{y}(t)$, but $G(s)=\delta(s)$ has its own merits. We will see later that for the quantum simulation protocol that the choice $G(s)=\delta(s)$ corresponds to more precise initial state preparation with no final measurements required, whereas $G(s)=1$ requires a precise final measurement step. \\

We note that in this case, while $t \in [0, T]$ where $T$ is some final time $T$, the \enquote{clock parameter} $s \in \mathbb{R}$ tracks the time. The clock parameter has a much larger range than $t \in [0,T]$ to give flexibility on where we want the clock to start. If we initialise $\vect{w}(0, s)$ using  $G(s)=\delta(s)$, this is equivalent to the clock starting at $s=0=t_0$, precisely. \\

When implementing the algorithms, one often needs to approximate the dirac $\delta$-function by a bounded and narrowly supported function $\delta_\omega$ which becomes zero or vanishingly small outside a domain of $O(\omega)$, for $\omega\ll 1$, namely, we consider
\begin{align}\label{y-omega}
    \vect{y}_{\omega}(t)=\int_{-\infty}^{\infty}\delta_{\omega}(s-t)\mathcal{U}_{s,s-t}\vect{y}_0 ds.
\end{align}
As conventionally done, we choose $\delta_\omega$ to be smooth and to satisfy, for $x\in\mathbb{R}^1$,
\begin{equation} \label{eq:deltaomega2}
\delta_\omega (x) = 0 \quad {\text {if} } \quad |x|> \omega; \quad
\int_{|x|\le \omega} \delta_\omega(x) \, dx =1.
\end{equation}
One usually  approximates $\delta_\omega$ by the form
\begin{equation} \label{eq:deltaomega}
\delta_\omega(x) =
\begin{cases} \frac{1}{\omega}\beta(x/\omega) \quad 
           & |x|\le \omega; \\
            0 \quad 
            &|x|> \omega,
\end{cases}
\end{equation}
where typical choices of $\beta(x)$ include
$\beta(x)=1-|\beta|$ and $\beta(x)=\frac{1}{2}(1+\cos(\pi x))$ \cite{EngTorn}. 
In a discrete-variable formulation, one can choose $\omega=mh$ where $m$ is the number of mesh points within the support of $\delta_\omega$, and $h$ is the grid size. 
For a continuous-variable formulation, one can choose $\delta_{\omega}$ to be a Gaussian function 
\begin{align*} 
    \delta_{\omega}(s-t)=\frac{1}{\sqrt{2\pi \omega^2}}e^{-(s-t)^2/(2\omega^2)},
\end{align*}
which can give good approximation to $\vect{y}(t)$ if the width $\omega$ is small. 

A standard analysis can easily show the following:
\begin{lemma} \label{lemma::error_classical}
Let $\vect{y}_{\omega}(t)$ be given by Eq.~\eqref{y-omega} and $\vect{y}(t)$ be given by Eq.~\eqref{ODE}. Assume $\|\vect{y}_0\|=1$, and the normalised states are expressed as $|y_{\omega}(t)\rangle=\vect{y}_{\omega}(t)/\|\vect{y}_{\omega}(t)\|$, $|y(t)\rangle=\vect{y}(t)/\|\vect{y}(t)\|=\vect{y}(t)$ where $\|.\|$ is the $l_2$-norm. Furthermore, 
Assume that 
\begin{enumerate}[label=(\roman*)]
\item $\delta_w(\cdot)$ is a probability distribution with mean $\mu = o(\omega)$ and second moment $\omega^2$ (presumably $\omega \ll 1$); 

\item the Hamiltonian $\vect{H}(t)$ is continuously differentiable with respect to $t$.
\end{enumerate}

Then, in the small $\omega$ limit, the quantum fidelity $\text{Fid} (. \, , .)$ between the ideal $|y(t)\rangle$ and the approximated state $|y_{\omega}(t)\rangle$ is
\begin{align*}
\text{Fid} \big(|y_{\omega}(t)\rangle, |y(t)\rangle \big) & := |\langle y_{\omega}(t)|y(t)\rangle|^2 = 1 - \omega^2 \mathsf{C}_{\text{R}} + \order{\omega^3},\\
\end{align*}
where $\mathsf{C}_{\text{R}}\ge 0$ is given by 
\begin{align*}
 \mathsf{C}_{\text{R}} = \expval{\vect{H}^2(t)}{y(t)}+\expval{\vect{H}^2(0)}{y(0)} - 2 \Re\Big(\bra{y(t)} \vect{H}(t) \mathcal{U}_{t,0} \vect{H}(0) \ket{y(0)}\Big).
\end{align*}
\end{lemma}
\begin{proof}
    See Appendix~\ref{appendix::proof_classical}. 
\end{proof}

\begin{remark} The mean $\mu$ characterises the bias and the second moment characterises how dispersive the distribution is. 
Usually when one regularises the Dirac-delta function one uses even functions thus $\mu=0$, as those examples given above. We keep $\mu$ more general, since the preparation of quantum state may not guarantee that $\mu=0$ so instead here we assumer  $\mu \sim 0$.
\end{remark}

\begin{example}[Hamiltonian commuting at different times]
If $\vect{H}(t) = \lambda(t) \vect{h}$, where $\vect{h}$ is a time-independent Hamiltonian and $\lambda(t)$ is a polynomial. Then $\vect{y}(t)=\exp\big(-i\int_{0}^{t} ds \,\lambda(s) \vect{h}\big) \vect{y}_0$, and we can readily compute that 
\begin{align*}
\mathsf{C}_{\text{R}} = (\lambda(t) - \lambda(0))^2\expval{\vect{h}^2}{y_0}.
\end{align*}
Clearly, the approximation error arising from $\omega>0$ mainly depends on the Hamiltonian at the initial and the final times.
\end{example}

\section{Quantum simulation of a linear non-autonomous system using a time-independent Hamiltonian} \label{sec:quantumsimintro}
We distinguish two kinds of non-autonomous dynamical systems. One obeys unitary dynamics 
\begin{align} \label{eq:timeschr2}
\begin{aligned}
    & \frac{d\vect{y}(t)}{dt}=-i\vect{H}(t)\vect{u}(t), \qquad \vect{H}(t)=\vect{H}^{\dagger}(t), \\
    & \vect{y}(0)=\vect{y}_0.
 \end{aligned}
\end{align}
and can describe closed quantum systems. The second more general dynamical system does not obey unitary dynamics
\begin{align} \label{eq:nonunitary}
\begin{aligned}
    & \frac{d\vect{u}(t)}{dt}=-i\vect{A}(t)\vect{u}(t), \qquad \vect{A}(t) \neq \vect{A}^{\dagger}(t),  \\
    & \vect{u}(0)=\vect{u}_0.
  \end{aligned}
\end{align}
For systems obeying unitary dynamics, we only need to apply the formalism in Section~\ref{sec:nonautotransformation} for transforming a non-autonomous system into an autonomous system with unitary dynamics. For systems obeying non-unitary dynamics, the Sch\"odingerisation procedure is required to turn Eq.~\eqref{eq:nonunitary} into Eq.~\eqref{eq:timeschr2}. We will look at the quantum simulation of these cases in the following. 

\subsection{Non-autonomous unitary dynamics}
We are given a system $\vect{y}(t)$ evolving under unitary dynamics in Eq.~\eqref{eq:timeschr2}, so no Schr\"odingerisation is required. 
To apply the formulation in Theorem~\ref{thm:one},  we first define $\vect{\bar{w}}(t)\equiv \int ds\, \vect{w}(t,s)|s\rangle$ and our aim is to first prepare the quantum state $|\bar{w}(t)\rangle=\vect{\bar{w}}(t)/\|\vect{\bar{w}}(t)\|$ using a \textit{time-independent} Hamiltonian $\vect{\bar{H}}$ and from this state to obtain $|y(t)\rangle \equiv \vect{y}(t)/\|\vect{y}(t)\|$. First we introduce the position operator $\hat{s}$ and the momentum operator $\hat{p}_s$ where $[\hat{s}, \hat{p}_s]=i\mathbf{1}_s$.  Using the standard procedure \cite{2023analog} we can make the replacement $s \rightarrow \hat{s}$ and $\partial/\partial s \rightarrow i\hat{p}_s$. Then we can rewrite the linear autonomous system Eq.~\eqref{w-eqn} in Theorem~\ref{thm:one} as 
\begin{align} \label{eq:wbaroriginal1}
    \frac{d \vect{\bar{w}}}{dt}=-i\vect{\bar{H}}\vect{\bar{w}}, \qquad \vect{\bar{H}}=\mathbf{1}\otimes\hat{p}_s + \vect{H}(\hat{s})=\vect{\bar{H}}^{\dagger}, \qquad \vect{\bar{w}}(0)=\vect{y}_0 \int  ds\, G(s) |s\rangle, 
\end{align}
where the Hamiltonian $\vect{\bar{H}}$ is \textit{time-independent}. In fact, we can even generalise beyond the pure initial state $|\bar{w}(0)\rangle$ to a more general, possibly mixed,  initial state $\sigma=|y_0\rangle \langle y_0| \otimes \rho_0$, which we will prove in Theorem~\ref{thm:two} below. We note that even though in Theorem~\ref{thm:two} (and later Theorem~\ref{thm:three}) $s$ is continuous, we can also discretise $s$ and we will see this in Sections~\ref{sec:qubitformalism} and~\ref{sec:hybridformalism}.

\begin{theorem} \label{thm:two} 
    Given the solution in Eq.~\eqref{ODE-solution} to the linear non-autonomous dynamical system in Eq.~\eqref{eq:timeschr2} with initial condition $\vect{y}_0$.
    Then $|y(t)\rangle$
can be simulated via unitary evolution with respect to the \textit{time-independent} Hamiltonian $\vect{\bar{H}}$ in the following way. We define the quantum state $\sigma(t)$ evolving according to
\begin{equation} 
\label{eq:sigmatseven}
\begin{aligned}
    & \frac{d \sigma(t)}{dt}=-i [\vect{\bar{H}}, \sigma(t)], \qquad  \vect{\bar{H}}=\mathbf{1}\otimes \hat{p}_s+\vect{H}(\hat{s}), \\
    &\sigma(0)=|y_0\rangle \langle y_0|\otimes \rho_0,\qquad \rho_0=\iint ds\, ds'\, g(s, s')|s\rangle \langle s'|, \qquad \int ds\, g(s,s)=1,
\end{aligned}
\end{equation}
   where $\rho_0=\rho^{\dagger}_0$ and $\rho_0$ is also positive semidefinite. With the choice of $\rho_0$ where $g(s, s)=\delta(s)$, then $|y(t)\rangle\langle y(t)|=\text{Tr}_s(\sigma(t))$ where the trace is over the $|s\rangle$ mode. Alternatively, with the choice of measuring $\sigma(t)$ in the mode $|s=t\rangle$, then $|y(t)\rangle\langle y(t)|=\text{Tr}((\mathbf{1} \otimes |s=t\rangle \langle s=t|) \sigma(t))/g(0,0)$, which means retrieving $|y(t)\rangle$ with probability $g(0,0)$.
\end{theorem}
\begin{proof}
    For simplicity we can assume $\|\vect{y}_0\|=1=\|\vect{u}_0\|$ without losing generality. Using Theorem~\ref{thm:one}, we can write the solution of Eq.~\eqref{eq:wbaroriginal1} as
    \begin{align*}
        \vect{\bar{w}}(t)=e^{-i\vect{\bar{H}}t}\vect{\bar{w}}(0)=e^{-i\vect{\bar{H}}t}\vect{y}_0\int ds\, G(s)|s\rangle=\int ds\, \mathcal{U}_{s, s-t} \vect{y}_0 G(s-t)|s\rangle.
    \end{align*}
    Similarly when $G(s) \in \mathbb{R}$
    \begin{align*}
        \vect{y}^{\dagger}_0\int ds\, \langle s|G(s) e^{i\vect{\bar{H}}t}=\int \vect{y}^{\dagger}_0 \mathcal{U}^{\dagger}_{s,s-t}G(s-t)\langle s|. 
    \end{align*}
    We can apply this to solving Eq.~\eqref{eq:sigmatseven}
\begin{align*}
    \sigma(t)=e^{-i\vect{\bar{H}}t}\sigma(0) e^{i\vect{\bar{H}}t}, \qquad \sigma(0)=\vect{y}_0 \vect{y}_0^T\otimes \rho_0=\vect{y}_0 \vect{y}_0^T \iint ds ds'g(s,s')|s\rangle \langle s'|,
\end{align*}
where we note that $g(s,s) \in \mathbb{R}$ since $\rho_0=\rho^{\dagger}_0$. 
 Thus
 \begin{align*}
    & \sigma(t)=e^{-i\vect{\bar{H}}t}\sigma(0) e^{i\vect{\bar{H}}t}=e^{-i\vect{\bar{H}}t}\left(\vect{y}_0\vect{y}^{\dagger}_0 \iint ds\, ds'\, g(s,s')|s\rangle \langle s'| \right)e^{i\vect{\bar{H}}t} \\
     &=\left(\iint ds\, ds'\, \mathcal{U}_{s,s-t}\vect{y}_0 \vect{y}^{\dagger}_0 g(s-t, s')|s\rangle \langle s'| \right)e^{i\vect{\bar{H}}t}=\iint ds \,ds'\, g(s-t, s'-t) \mathcal{U}_{s, s-t} \vect{y}_0 \vect{y}^{\dagger}_0  \mathcal{U}^{\dagger}_{s', s'-t} |s\rangle\langle s'|.
 \end{align*}
 For the first protocol of tracing out the $|s\rangle$ mode in $\sigma(t)$, we have
 \begin{equation}
 \label{eqn::gamma}
    \gamma(t)=\text{Tr}_s(\sigma(t))=\int ds \,g(s-t, s-t)\mathcal{U}_{s, s-t}\vect{y}_0\vect{y}^T_0\mathcal{U}^{\dagger}_{s, s-t} .
 \end{equation}
 In the case $g(s-t, s-t)=\delta(s-t)$, then clearly $\gamma(t)=\mathcal{U}_{t,0}\vect{y}_0\vect{y}^{\dagger}_0\mathcal{U}^{\dagger}_{t,0}=|y(t)\rangle \langle y(t)|$. This is the simplest protocol, since not measurements are required at the end, and we simply just need to throw away the final clock register. \\

 For an alternative protocol of measuring $\sigma(t)$ in the state $|s=t\rangle$, then
 \begin{align*}
     \text{Tr}\big((\mathbf{1}\otimes |s=t\rangle\langle s=t|)\sigma(t)\big)=g(0,0)\mathcal{U}_{t,0}\vect{y}_0\vect{y}^{\dagger}_0\mathcal{U}^{\dagger}_{t,0}=g(0,0)|y(t)\rangle \langle y(t)|.
 \end{align*}
 
\end{proof}

\subsection{Non-autonomous non-unitary dynamics}
For a general linear dynamical system
\begin{equation} \label{eq:uequation}
    \frac{d \vect{u}}{d t}=-i\vect{A}(t)\vect{u}, \qquad \vect{A}(t) \neq \vect{A}^{\dagger}(t) , 
\end{equation}
so we use the following decomposition
\begin{equation*}
\vect{A}(t)=\vect{A}_1(t)-i \vect{A}_2(t), \qquad {\text {where}} \quad \vect{A}_1(t)=\frac{1}{2}(\vect{A}(t)+\vect{A}^{\dagger}(t))=\vect{A}^{\dagger}_1(t), \qquad \vect{A}_2(t)=\frac{i}{2}(\vect{A}(t)-\vect{A}^{\dagger}(t))=\vect{A}^{\dagger}_2(t).
\end{equation*}
 To turn this system into Eq.~\eqref{eq:timeschr2}, we can apply Schr\"odingerisation \cite{jin2212quantum, jin2022quantum, 2023analog}.  The first step  is the application of a warped phase transformation to $\vect{u}$ by defining a new parameter $\xi \in \mathbb{R}$, where for $\xi>0$, and let $\vect{v}(t, \xi)=\exp(-\xi)\vect{u}(t)$. We then  extend evenly to $\xi<0$, see details in \cite{jin2022quantum}. Then Eq.~\eqref{eq:uequation} becomes
\begin{align*}
    \frac{\partial \vect{v}}{\partial t}=\vect{A}_2(t)\frac{\partial \vect{v}}{\partial \xi}-i \vect{A}_1(t) \vect{v}, \qquad \vect{v}(0, \xi)=e^{-|\xi|}\vect{u}_0.
\end{align*}
Taking the Fourier transform of $\vect{v}(t)$ with respect to $\xi$, $\vect{v} \rightarrow \tilde{\vect{v}}$, gives

\begin{align}
    \frac{\partial \tilde{\vect{v}}}{\partial t}=-i( \eta \vect{A}_2(t)+\vect{A}_1(t)) \tilde{\vect{v}}=-i\vect{H}_{\eta}(t)\tilde{\vect{v}} ,
\end{align}
where $\eta \in (-\infty, \infty)$ is the Fourier mode. Now the new Hamiltonian of the system is $\vect{H}_{\eta}(t)=\eta \vect{A}_2(t)+\vect{A}_1(t)=\vect{H}_{\eta}^{\dagger}(t)$. The only difference is that now this is a system of Schr\"odinger type equations, one for each $\eta$. It can be represented  by augmenting $\tilde{\vect{v}}$ by a single continuous-variable mode $\vect{\eta}$, and $\hat{\eta}$ acts on $\vect{\eta}$ by $\hat{\eta} \vect{\eta}=\eta \vect{\eta}$. Note that we have a corresponding conjugate operator $\hat{\xi}$ where $[\hat{\xi}, \hat{\eta}]=iI_{\eta}$ and $\hat{\xi}$ acts on mode $\vect{\xi}$ by $\hat{\xi}\vect{\xi}=\xi \vect{\xi}$. Since $\eta, \xi$ are conjugate variables, we can write $\langle \eta|\xi\rangle=\exp(-i\xi \eta)$. Defining $\vect{y}(t)=\int \tilde{\vect{v}}(t, \eta)|\eta\rangle d\eta$, we then arrive at Eq.~\eqref{eq:timeschr} 
\begin{align} \label{eq:y}
\begin{aligned}
    & \frac{d \vect{y}}{d t}=-i \vect{H}(t)\vect{y}, \qquad \vect{H}(t)=\hat{\eta} \otimes \vect{A}_2(t)+I_{\eta} \otimes \vect{A}_1(t), \qquad \vect{H}(t)=\vect{H}^{\dagger}(t), \\
    & \vect{y}_0=|\Xi\rangle \vect{u}_0, \qquad |\Xi\rangle=\int \frac{2}{1+\eta^2}|\eta\rangle d\eta=\int e^{-|\xi|}|\xi\rangle d\xi,
    \end{aligned}
\end{align}
where $\vect{H}(t)$ is the time-dependent Hamiltonian and $I_{\eta}$ is the identity operation acting on the ancillary $\vect{\eta}$ mode. To recover $\vect{u}(t)$ from $\vect{y}(t)$, we just need to apply an inverse Fourier transform to $\vect{y}(t)$ with respect to $\eta$ to get $\vect{v}(t, \xi)$ and then to recover $\vect{u}(t)$ by
\begin{align*}
    \vect{u}(t)=\int_{0}^{\infty} \vect{v}(t, \xi) d\xi.
\end{align*}

However, Eq.~\eqref{eq:y} is a non-autonomous system, thus it may not be convenient in all formalisms to recover $|y(t)\rangle$ using the time-dependent $\vect{H}(t)$. Since Eq.~\eqref{eq:y} is now in the form of Eq.~\eqref{eq:timeschr2}, we can apply Theorem~\ref{thm:two} to turn Eq.~\eqref{eq:y} into an autonomous system with unitary dynamics and recover $|y(t)\rangle$ using a time-independent Hamiltonian $\vect{\bar{H}}$. Then the state $|u(t)\rangle \equiv \vect{u}(t)/\|\vect{u}(t)\|$ can be recovered from $|y(t)\rangle$ through a measurement procedure. See Figure~\ref{fig:perfect} for a schematic diagram. 

\begin{theorem} \label{thm:three}
    The linear non-autonomous dynamical system in Eq. \eqref{eq:nonunitary}
has a corresponding non-autonomous system $\vect{y}(t)$ that evolves with respect to the unitary $\mathcal{U}_{t,0}=\mathcal{T}\exp(-i\int_0^t\vect{H}(\tau) d\tau)$ and obeys Eq.~\eqref{eq:y} with initial condition $\vect{y}_0=|\Xi\rangle \vect{u}_0$, where $|\Xi\rangle=\int \exp(-|\xi|)|\xi\rangle d\xi$. Then $|u(t)\rangle$
can be simulated via unitary evolution with respect to the \textit{time-independent} Hamiltonian $\vect{\bar{H}}$ in the following way. We define the quantum state $\sigma(t)$ evolving according to
\begin{align} \label{eq:sigmat}
\begin{aligned}
    & \frac{d \sigma(t)}{dt}=-i [\vect{\bar{H}}, \sigma(t)], \qquad \vect{\bar{H}}=I_{\eta}\otimes \mathbf{1}\otimes \hat{p}_s+\vect{H}(\hat{s})=I_{\eta}\otimes \mathbf{1}\otimes \hat{p}_s +\hat{\eta}\otimes \vect{A}_2(\hat{s})+I_{\eta}\otimes \vect{A}_1(\hat{s}),\\ 
    &\sigma(0)=|y_0\rangle \langle y_0|\otimes \rho_0, \qquad |y_0\rangle=\frac{\vect{y}_0}{\|\vect{y}_0\|},\qquad \rho_0=\iint ds ds' g(s, s')|s\rangle \langle s'|, \qquad \iint ds g(s,s)=1,
\end{aligned}
\end{align}
where $\rho_0=\rho^{\dagger}_0$ and $\rho_0$ is also positive semidefinite.
   With the choice of $\rho_0$ where $g(s, s)=\delta(s)$, then $|y(t)\rangle\langle y(t)|=\text{Tr}_s(\sigma(t))$ where the trace is over the $|s\rangle$ mode. Alternatively, with the choice of measuring $\sigma(t)$ in the mode $|s=t\rangle$, then $|y(t)\rangle\langle y(t)|=\text{Tr}((I_{\eta} \otimes \mathbf{1} \otimes |s=t\rangle \langle s=t|) \sigma(t))/g(0,0)$, which means retrieving $|y(t)\rangle$ with probability $g(0,0)$. Given any measurement $\hat{P}_{>0}=\int_0^{\infty}f(\xi)|\xi\rangle \xi|d\xi$, then $|u(t)\rangle$ can be retrieved from $|y(t)\rangle$ using $\hat{P}_{>0}|y(t)\rangle \propto |u(t)\rangle$ with success probability $O(\int_0^{\infty} f(\xi) e^{-\xi} d \xi(\|\vect{u}(t)\|/\|\vect{u}_0\|)^2)$.
\end{theorem}

\begin{proof}
For simplicity we can assume $\|\vect{y}_0\|=1=\|\vect{u}_0\|$ without losing generality. The first step associating $\vect{u}(t)$ in Eq.~\eqref{eq:nonunitary} with $\vect{y}(t)$ obeying Eq.~\eqref{eq:y} is just the Schr\"odingerisation procedure, where $\vect{y}(t)$ evolves with respect to the time-dependent Hamiltonian $\vect{H}(t)=\hat{\eta}\otimes \vect{A}_2(t)+I_{\eta}\otimes \vect{A}_1(t)$ where the corresponding unitary evolution is denoted $\mathcal{U}_{t,0}$. \\

With $\vect{H}(t)$ given in Eq.~\eqref{eq:sigmat},  $\vect{H}(\hat{s})=\hat{\eta}\otimes \vect{A}_2(\hat{s})+I_{\eta}\otimes \vect{A}_1(\hat{s})$. Therefore, applying Theorem~\ref{thm:two}, we have the time-independent Hamiltonian as 
\begin{align*}
    \vect{\bar{H}}=I_{\eta}\otimes \mathbf{1} \otimes \hat{p}_s+\vect{H}(\hat{s})=I_{\eta}\otimes \mathbf{1} \otimes \hat{p}_s+\hat{\eta}\otimes \vect{A}_2(\hat{s})+I_{\eta}\otimes \vect{A}_1(\hat{s}).
\end{align*}
To obtain $|y(t)\rangle$ from $\sigma(t)$, the proof is identical to Theorem~\ref{thm:two}. \\

To retrieve $|u(t)\rangle$ from $|y(t)\rangle$, a simple measurement procedure is sufficient, outlined in more detail in \cite{2023analog, jin2022quantum}. This is a simple projective measurement $\hat{P}=\int_0^{\infty} d\xi |\xi\rangle \langle \xi|$. This can be shown to be equivalent to accepting the remaining state, which is $|u(t)\rangle$, if we postselect on $\xi>0$. We can also generalise this to imperfect projective measurements in the $|\xi\rangle$ basis and define $\hat{P}_{\text{imp}}=\int_0^{\infty} f(\xi)|\xi\rangle \langle \xi|$ where $f(\xi)$ models the imperfection in the detector. For instance, it could be a top-hat function or a Gaussian function  with the width or standard deviation denoting the precision of the detector. In this case, after measurement, the probability of retrieving $|u(t)\rangle$ is now $(\int_0^{\infty} f(\xi)e^{-\xi} d\xi \|\vect{u}(t)\|/\|\vect{u}_0\|)^2$ \cite{2023analog}. 
\end{proof}
\begin{figure}[ht] 
\centering
\includegraphics[width=12cm]{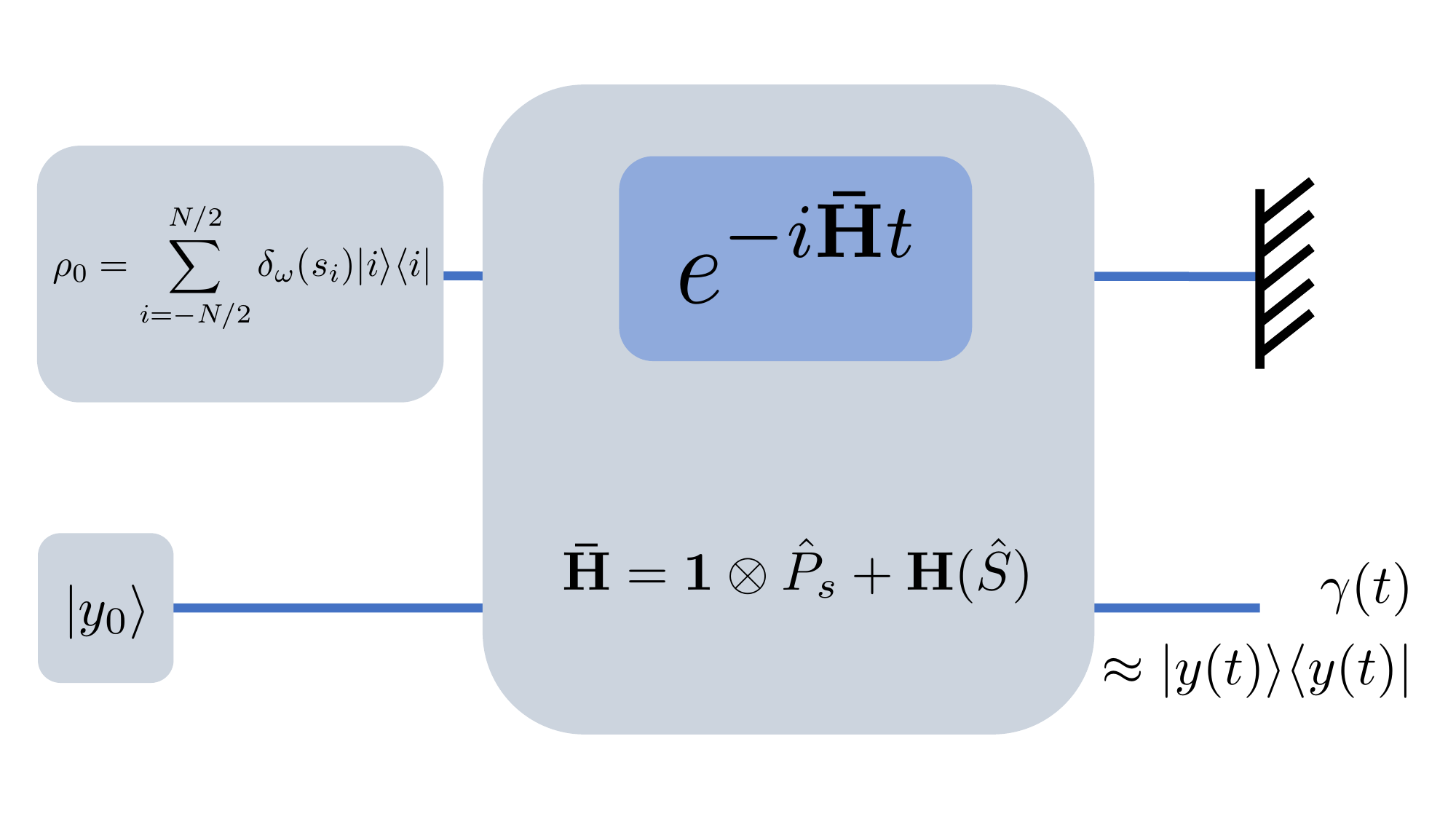} 
\caption{\textit{Time-independent Hamiltonian simulation for non-autonomous quantum dynamics in qubit or hybrid setting}. Here we use the simpler first protocol in Theorem~\ref{thm:two}, where we have a time-dependent Hamiltonian $\vect{H}(t)$. This is in principle suitable for fully continuous-variable systems, but for illustration we assume that we use $\log_2(N)$ qubits for the clock mode. Here we can begin with a mixed quantum initial state $\rho_0$ which is simple to prepare and we are assumed we are given $|y_0\rangle$. Then the system can evolve under a time-independent Hamiltonian $\vect{\bar{H}}$. Retrieval of an approximation to $|y(t)\rangle$ is very simple, where we simply throw away the clock mode at the end.}
\label{fig:quantumdynamics}. 
\end{figure}

\begin{figure}[ht] 
\centering
\includegraphics[width=12cm]{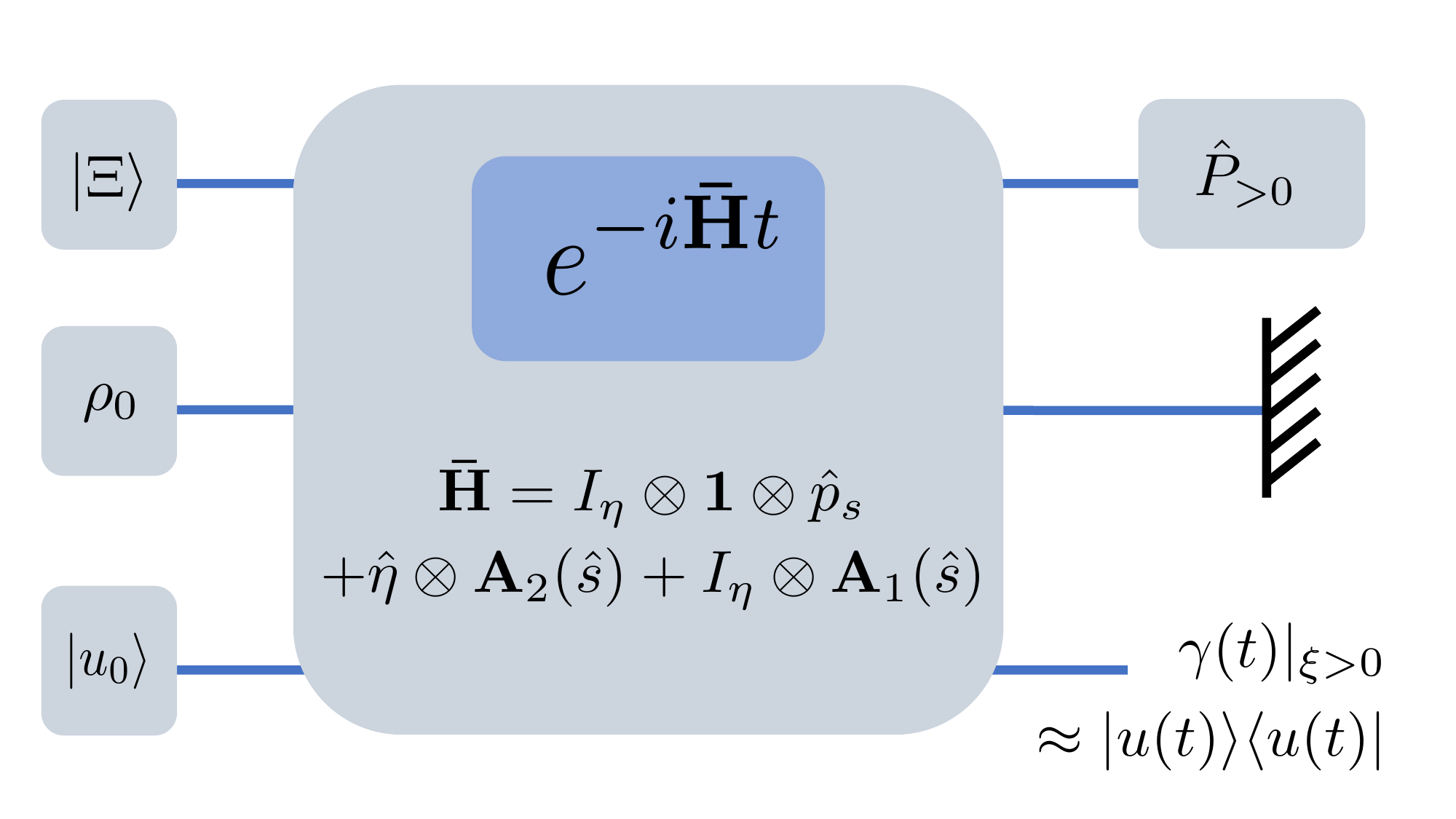} 
\caption{\textit{Time-independent Hamiltonian simulation for general non-autonomous linear PDEs}. Here we use the simpler first protocol in Theorem~\ref{thm:three}. In the ideal initial state preparation, we have the ancillary initial state as the $s=0$ eigenstate $\rho_0=|s=0\rangle \langle s=0|$. In this scenario, after the evolution of the initial state with respect to the time-independent Hamiltonian $\vect{\bar{H}}$, we perform the measurement $\hat{P}_{>0}$ postselected on $\xi>0$, and we retrieve the state $|\bar{w}(t)\rangle=\vect{\bar{w}}(t)/\|\vect{\bar{w}}(t)\|$ with probability $\|\vect{\bar{w}}(t)\|^2/\|\vect{\bar{w}}(0)\|^2$ where $\vect{\bar{w}}(t)=\exp(-i\vect{\bar{H}}t)\vect{\bar{w}}(0)$, with $\vect{\bar{w}}(0)=\vect{y}_0|s=0\rangle$. Then $|\bar{w}(t)\rangle=|s=t\rangle |u(t)\rangle$ and by tracing out the ancilla state $ |s=t\rangle$, we retrieve $|u(t)\rangle$ exactly. However, in general $|s=0\rangle$ is an ideal state preparation that cannot be achieved in reality and we can use an approximation. We can also prepare more general and possibly mixed state $\rho_0=\iint g(s,s')|s\rangle \langle s'| ds ds'$, where the output state $\gamma(t)=\text{Tr}_s(\sigma(t)) \approx |y(t)\rangle \langle y(t)|$ when $g(s,s)$ is close enough to $\delta(s)$.}
\label{fig:perfect}. 
\end{figure}
For applications to different ODEs and PDEs, we see that we only need the corresponding $\vect{A}(t)$ for the problem \cite{2023analog}, and we can use either Theorem~\ref{thm:two} when $\vect{A}(t)=\vect{A}^{\dagger}(t)$ or Theorem~\ref{thm:three} when $\vect{A}(t) \neq \vect{A}^{\dagger}(t)$. It can be applied to both linear homogeneous and inhomogeneous PDEs and also a system of nonlinear ODEs and classes of nonlinear PDEs. \\

\subsection{Preparing $|u(t)\rangle$ with an (imperfect) auxiliary clock register}
The first protocol in Theorems~\ref{thm:two} and~\ref{thm:three} that prepares $|y(t)\rangle$ by tracing out the clock register $|s\rangle$ entirely is generally simpler for implementation, since no measurement process is required at the end. An important observation from Theorems~\ref{thm:two} and~\ref{thm:three} is that only the diagonal components $g(s,s)$ of the initial ancilla state $\rho_0$ matter for retrieving $|y(t)\rangle$ -- and subsequently $|u(t)\rangle$ -- rather than any off-diagonal term $g(s,s')$ for $s \neq s'$. This means it is not the quantum coherence that acts as a resource but rather the `classical precision' in the clock register $|s\rangle$ through using $g(s,s)=\delta(s)$. Since $g(s,s)=\delta(s)$ refers to a clock register with infinite precision, this is not possible to implement, and we will instead prepare physically realisable states, so we require $\delta \rightarrow \delta_{\omega}$. In an implementation, we can also prepare a mixed state $\rho_0 = \int ds \delta_{\omega}(s)|s\rangle \langle s|$ instead and this achieves the same precision as when using a pure state $|\psi_0\rangle=\int ds \, \sqrt{\delta_{\omega}(s)}|s\rangle$, since we retrieve the same $\gamma(t)=\text{Tr}_s(\sigma(t))$ state in cases. \\

We want to bound the quantum fidelity $\text{Fid}(. \, , .)$ between the target $|u(t)\rangle$ and the true output state $\gamma(t)=\text{Tr}_s(\sigma(t))$, using the first protocol in Theorems~\ref{thm:two} and~\ref{thm:three}. For simplicity, we will not include the part requiring Schr\"odingerisation and work with $\vect{A}(t)=\vect{H}(t)$, so $|u(t)\rangle=|y(t)\rangle$ and 
\begin{align*}
    \text{Fid}(\gamma(t), |y(t)\rangle)&=\text{Tr}\ (\gamma(t)|y(t)\rangle \langle y(t)|)=\int ds\, g(s-t, s-t)|\langle y(t)|\mathcal{U}_{s, s-t}|y_0\rangle|^2 \geq 1-\zeta.
\end{align*}
When $g(s,s)=\delta(s)$, then clearly $\text{Fid}(\gamma(t), |y(t)\rangle)=1$. However, if we can only access an approximation, $g(s,s)=\delta_{\omega}(s)$. As we see from Lemma~\ref{lemma::error_quantum}, this will give $1-\text{Fid}(\gamma(t), |y(t)\rangle) \approx O(\omega^2)$, thus we need to choose widths $\omega \lesssim O(\sqrt{\zeta})$. \\

The second protocol in Theorems~\ref{thm:two} and~\ref{thm:three} that require a final projective measurement in $|s=t \rangle$ doesn't exist in the limit of $g(0,0)=1$. Instead, suppose we prepare a thermal (Gaussian) state with 
\begin{align*}
g(s,s) &=\exp(-s^2/2\omega^2)/\sqrt{2\pi \omega^2},\ \forall s \text{ and } g(s,s') = 0,\  \forall s\neq s'.
\end{align*}
Then by measuring the last register in $|s=t\rangle$, we obtain $|y(t)\rangle$ with probability $g(0,0)=1/\sqrt{2 \pi \omega^2}$, which is only possible with perfect final measurement. In this case, although a larger $\omega$ corresponding to a higher temperature state which might be simpler to prepare in some cases, a final more precise measurement is still required. Furthermore, the larger the $\omega$, the smaller $g(0,0)$ so the final probability of retrieval of $|y(t)\rangle$ and hence $|u(t)\rangle$ will be multiplied by the factor $\omega$. Thus, in most cases where there is not much difference in resource in the preparation of $\rho_0$, which does not require any quantum coherence, the first protocol is preferred. \\

\begin{lemma}
\label{lemma::error_quantum}
Let the reduced density matrix $\gamma_{\omega}(t)=\text{Tr}_s(\sigma(t))$, where $\sigma(t)$ is defined by either Theorem~\ref{thm:two} or~\ref{thm:three} and let $g(s,s)=\delta_{\omega}(s)$. Then 
\begin{align*}
 \gamma_{\omega}(t) = \int \delta_{\omega}(s-t) \mathcal{U}_{s,s-t} \dyad{\vect{y}_0} \mathcal{U}_{s,s-t}^\dagger\ ds.
\end{align*}
The quantum fidelity and trace-distance between the approximated state $\gamma_{\omega}(t)$ and the target state $|y(t)\rangle$ are given by 
 \begin{align*}
& \text{Fid}\big(\gamma_{\omega}(t), |y(t)\rangle \big) = 
 1 - \mathsf{C} \omega^2  + o(\omega^2)\\
& \norm{\gamma_{\omega}(t) - \dyad{y(t)}}_{\tr} \le \sqrt{1- \text{Fid}\big(\gamma_{\omega}(t), \dyad{y(t)} \big)} = \sqrt{\mathsf{C}}\ \omega + o(\omega).
 \end{align*}
 and the prefactor
 \begin{align*}
 \mathsf{C} &= \expval{\vect{H}^2(t)}  - 2\text{Re}\bra{y(t)} \vect{H}(t) \mathcal{U}_{t,0} \vect{H}(0) \ket{y_0} +  \expval{\vect{H}^2(0)}  - \Big(\expval{\vect{H}(t)} - \expval{\vect{H}(0)}\Big)^2  \le \mathsf{C}_{\text{R}}.
 \end{align*}
\end{lemma}
\begin{proof}
    See Appendix~\ref{appendix::proof_quantum}. 
\end{proof}
We remark that  the above error analysis implicitly assumes that the case $\mu \gg \order{\omega}$ is disregarded, where $\mu$ is the bias. This case can be avoided with most state preparations.  

\begin{example}
If $\vect{H}(t) = \lambda(t) \vect{h}$, where $\vect{h}$ is a time-independent Hamiltonian and $\lambda(t)$ is a polynomial, we can easily show
\begin{align*}
\mathsf{C} =& (\lambda(t) - \lambda(0))^2 \Big(\expval{\vect{h}^2}{y_0} - (\expval{\vect{h}}{y_0})^2\Big).
\end{align*}
\end{example}

\subsection{Qubit-based system} \label{sec:qubitformalism}
If $\vect{H}(t)$ acts on a finite-dimensional Hilbert space, this protocol can also be formulated entirely in terms of operations on qubits instead of qumodes. We begin with the linear non-autonomous unitary dynamical system
\begin{align} \label{eq:hsystem}
    \frac{d \vect{u}}{dt}=-i\vect{H}(t)\vect{u}, \qquad \vect{H}^{\dagger}=\vect{H}, \qquad \vect{u}_0=\vect{u}(0),
\end{align}
which can either represent a quantum dynamical system of interest or the already Schr\"odingerised equation from a non-Schr\"odinger-like PDE or dynamical system. We assume unitarity in Eq.~\eqref{eq:hsystem} since the discretisation of Schr\"odingerisation is already dealt with in \cite{jin2022quantum, jin2212quantum} so here we can focus on the effect of the protocol in going from a non-autonomous to an autonomous system. Since this is operating on a system of qubits, $|u(t)\rangle=\vect{u}(t)/\|\vect{u}(t)\|$ resides in finite-dimensional Hilbert space and we can write
\begin{align*}
    \vect{u}(t)=\sum_{j=1}^J u_j(t)|j\rangle .
\end{align*}
Here Eq.~\eqref{eq:hsystem} can represent a quantum dynamical system of $n=\log_2 J$ qubits or two-level systems. If Eq.~\eqref{eq:hsystem} represents a discretised $D$-dimensional linear PDE, then $J=M^D$, where $M$ is the numnber of mesh points in the spatial discretisation used for each spatial dimension.\\

To implement the dynamics in Theorem~\ref{thm:three} entirely on a system of qubits, we must choose a discretisation for $s$. We can choose a uniform mesh size $\Delta s=1/(N+1)\sim 1/N$ with grid points denoted $s_{-N/2}<\cdots<s_{N/2}$. We can also define a discretised momentum operator $\hat{P}_s$ that discretises $-i\partial/\partial s$ and a discretised position operator $\hat{S}$ that discretises $\hat{s}$. We will use the first protocol in Theorem~\ref{thm:three} where the ideal initial state $\rho_0$ has $g(s,s)=\delta(s)$. We need to use a discretised $\delta_\omega(s_i)$ instead, defined in Eq.~\eqref{eq:deltaomega}. Then we can define a discretised $\sigma(t)$ which is composed of $\log_2((N+1)J)\sim \log_2(NJ)$ qubits. This $\sigma(t)$ evolves under unitary evolution
\begin{align} \label{eq:sigmaqubitevolution}
    \sigma(t)=e^{-i\vect{\bar{H}}t}\sigma(0)e^{i\vect{\bar{H}}t}, \qquad  \vect{\bar{H}}=\mathbf{1}\otimes \hat{P}_s+ \vect{H}(\hat{S}) ,
\end{align}
with $\vect{\bar{H}}$ is now represented by a $(N+1)J \times (N+1)J$ Hermitian matrix given in Eq.~\eqref{eq:sigmat}. That the ancilla $|\eta\rangle$ mode here is not required since we started with unitary dynamics in Eq.~\eqref{eq:hsystem}. We can choose our state $\rho_0$ to be a mixed state with only diagonal components $\delta_{\omega}(s_i)$, so the initial state $\sigma(0)$ can also be discretised 
\begin{align} \label{eq:sigma0discrete}
    \sigma(0) \rightarrow \sum_{i=-N/2}^{N/2} \sum_{j, j'=1}^J \sigma_{i, j, j'}|ji\rangle \langle j'i|=\sum_{i=-N/2}^{N/2} \sum_{j, j'=1}^J u_j(0)u^*_{j'}(0)|j\rangle \langle j'| \otimes \delta_{\omega}(s_i) |i\rangle \langle i|.
    \end{align}
It is important to note that $t$ remains continuous throughout the evolution. Thus, if $\vect{\bar{H}}$ can be naturally prepared by some analogue quantum simulator, then the evolution in Eq.~\eqref{eq:sigmaqubitevolution} can be performed directly. \\

 Since $\vect{\bar{H}}$ is already a time-independent Hamiltonian simulation, we can now apply any conventional or improved quantum simulation algorithms for time-independent methods, without resorting to methods especially designed for time-dependent systems that require time-ordered queries or multiple measurements throughout the evolution. The first protocol in Theorem~\ref{thm:three} (see also Figure~\ref{fig:perfect}) requires no measurements at all to deal with the time-dependency of $\vect{H}(t)$. Note also that unlike previous qubit-based methods for time-dependent Hamiltonians where the Dyson series expansion requires that time evolves by discrete steps, here $t$ is a continuous quantity. In principle it is not necessary to break up $t$ into discrete time steps unless $\vect{\bar{H}}$ cannot be naturally realised in an analogue setting. \\

For instance, let $s_{\bar{H}}$ denote the sparsity of matrix $\vect{\bar{H}}$ and let $\|\vect{\bar{H}}\|_{\max}$ denote its max-norm (value of its largest entry in absolute value). Let the $(j,k)^{\text{th}}$ entry of $\vect{\bar{H}}$ be denoted $(\vect{\bar{H}})_{jk}$. A common set of black-boxes used in Hamiltonian simulation is known as sparse access. 
\begin{definition}
Sparse access to Hermitian matrix $\vect{\bar{H}}$ refers to two unitary black-boxes $O_M$ and $O_F$ such that $O_M|j\rangle|k\rangle|z\rangle=|j\rangle|k\rangle|z\oplus H_{jk}\rangle$ and $O_{F} |j\rangle|k\rangle=|j\rangle|F(j,k)\rangle$. 
Here the function $F$ takes the row index $j$ and a number $k=1,2,...,s_{\bar{H}}$ and outputs the column index of the $k^{\text{th}}$ non-zero elements in row $j$. 
\end{definition}
There are quantum simulation protocols for quantum simulation with respect to time-independent Hamiltonians $\vect{\bar{H}}$, i.e.,  $\exp(-i\vect{\bar{H}}t)$,  in terms of query complexity that scale linearly in $t$ \cite{low2019hamiltonian} using sparse access or linearly in $t$ up to logarithmic factors \cite{berry2015hamiltonian}.

\begin{lemma} \cite{berry2015hamiltonian} ~\label{eq:lemmasimulation}
Let $\tau=s_{\bar{H}} T\|\vect{\bar{H}}\|_{\text{max}}$. Then $\exp(-i\vect{\bar{H}}T)$ acting on $m_{\bar{H}}$ qubits can be simulated to within error $\epsilon$ with query complexity $\mathcal{O}(\tau \log (\tau/\epsilon)/(\log\log (\tau/\epsilon)))$ and gate complexity 
$\mathcal{O}\Big(\tau ( m_{\bar{H}} + \log^{2.5}(\tau/\epsilon))\log (\tau/\epsilon)/(\log\log (\tau/\epsilon))\Big)$.
 \end{lemma}
 We can directly apply this to our protocol. Since to approximate $|u(t)\rangle$ we only need to throw away the extra $|s_i\rangle$ registers, consisting of $\log_2 N$ qubits, after obtaining $\sigma(t)$, this will not incur any extra costs. The initial state preparation for $\sigma(0)$ involves the preparation of $|u(0)\rangle$ and the mixed state $\rho_0=\sum_{i=-N/2}^{N/2} \delta_{\omega}(s_i)|i\rangle \langle i|$, where we use the approximate discretised $\delta$-function $\delta_{\omega}$ such that $\sum_{i=-N/2}^{N/2} \delta_{\omega}(s_i)=1$. Since $|u(0)\rangle$ is problem dependent, we assume it is given. We note $\rho_0$ is a classical mixture of states $\{|i\rangle\}_{i=-N/2}^{N/2}$, where each $|i\rangle$ is selected with probability $\delta_{\omega}(s_i)$, which means that $\rho_0$ does not exploit any quantum coherence. Thus there is no quantum cost in the preparation of $\rho_0$ once we are given $\{|i\rangle\}_{i=-N/2}^{N/2}$.

\begin{theorem} \label{thm:four}
        Assuming sparse access to $\vect{\bar{H}}$, the query and gate complexities for obtaining the $J$-qubit state $|y(T)\rangle$ to precision $\epsilon$ are respectively given in Lemma~\ref{eq:lemmasimulation}, where $m_{\bar{H}}=\log_2(J/\epsilon)$ and  $\tau \lesssim O(\max_{t \in [0, T]} s(\vect{H}(t))(1/\epsilon+\max_{t \in [0, T]} \|\vect{H}(t)\|_{max})T)$. 
\end{theorem}
\begin{proof} Here we just need to be concerned with the cost in obtaining $\sigma(T)$ in Eq.~\eqref{eq:sigmaqubitevolution} with initial condition in Eq.~\eqref{eq:sigma0discrete}, because $\gamma(T)=\text{Tr}_s(\sigma(T)) \approx |y(t)\rangle \langle y(t)|$. In our protocol, it is clear that the total number of qubits required is the number required for the state $|u(t)\rangle$ plus that required to represent the discretised ancillary mode. Thus  $m_{\bar{H}}=\log_2(NJ)=\log_2(N)+\log_2(J)$. For example, if we want to relate the size $N$ to the 2-norm error $\epsilon$ of our state, the fidelity between the target state and the actual state $\text{Fid} \geq 1-\zeta$ where $\epsilon^2 \sim \zeta \sim \omega^2$, so $\omega \sim \epsilon$. Suppose we choose the width $\omega \sim 1/N$, which is the minimum width possible with the discretisation size $N$, then $\epsilon \sim 1/N$. So we can make the replacement $N \sim 1/\epsilon$ up to constants.\\

The sparsity of $\vect{\bar{H}}$ is the sparsity of $\mathbf{1}\otimes \hat{P}_s+\vect{H}(\hat{S})$. In the first-order upwind scheme, for instance, the sparsity of $\hat{P}_s$ is $2$, so the sparsity of $\vect{\bar{H}}$ is dominated by the sparsity of $\vect{H}(\hat{S})$, so $s_{\bar{H}} \approx s (\vect{H}(\hat{S}))$. Since $\hat{S}$ is a diagonal matrix, the sparsity of $\vect{H}(\hat{S})=\sum_{j=0}^{N/2}\vect{H}(s_i) \otimes |i\rangle \langle i|$ (where here $\vect{H}(s_i)=0$ for all $s_i<0$) is only due to the maximum sparsity of all the matrices $\vect{H}(t)$ for every $t \in [0, T]$. So we can rewrite $s_{\bar{H}} \approx \max_t s(\vect{H}(t))$. For the max-norm, we note that $\|\vect{\bar{H}}\|_{\max} \le {\|\vect{H}(\hat{S})\|_{\text{max}}+\|\hat{P}_s\|_{\text{max}}}$ where $\|\hat{P}_s\|_{\text{max}}=N$. Now $\|\vect{H}(\hat{S})\|_{max}=\max_{t} \|\vect{H}(t)\|_{max}$ is the maximum norm of $\vect{H}(t)$ over all time $t \in [0, T]$ involved in the protocol. 
Inserting into Lemma~\ref{eq:lemmasimulation} we have our result. 
\end{proof}
For Theorem~\ref{thm:four}, it's not necessary to only use Lemma~\ref{eq:lemmasimulation}, but any protocol for time-independent Hamiltonians can be used in its place. \\

\begin{remark}
    The quantities $s(\vect{H}(\hat{S}))$ and $\|\vect{H}(\hat{S})\|_{\text{max}}$ cannot be computed without explicit access to the forms of $\vect{H}(t)$. However, we can still makes some remarks on their scaling in different situations. For example,  when $\vect{H}(t)$ results from the discretisation of a PDE with a highest-order $K^{\text{th}}$ derivative, then $\|\vect{H}(t)\|_{\text{max}} \sim O(M^K)$ where  $M$ denotes the number of mesh points in the discretisation of each dimension of the PDE, while $\|\hat{P}_{s}\|_{\text{max}}=O(N)$ where $N$ is the number of mesh points in the discretisation of $s$.  Since $K \ge 1$ and typically $M \sim N$, then we see $\|\hat{P}_{s}\|_{\text{max}}\lesssim \|\vect{H}(\hat{S})\|_{\text{max}}$, which implies $\|\vect{\bar{H}}\|_{\text{max}}\|\sim \vect{H}(\hat{S})\|_{\text{max}}$. Thus the complexity in simulating the time-dependent case is dominated by $O(\max_{t \in [0, T]} s(\vect{H}(t))\max_{t \in [0, T]} \|\vect{H}(t)\|_{max})T)$. This means that the presence of the extra terms in the now time-independent Hamiltonian $\vect{\bar{H}}$ due to the ancillary qubits does not significantly contribute to the complexity. In the case where $\|\hat{P}_s\|_{\text{max}}$ dominates the $\|\vect{H}(\hat{S})\|_{\text{max}}$, the total query and gate complexity is then only of order $\tau \sim O(\max_{t \in [0, T]} s(\vect{H}(t))T/\epsilon)$.  
    \end{remark}

\subsection{Hybrid system} \label{sec:hybridformalism}
We can also consider continuous-variable discrete-variable hybrid protocols, where $\vect{\bar{H}}$ acts on both continuous-variable (qumode) or discrete-variable (qubit) degrees of freedom. The first scenario is if the clock register is kept as a continuous variable, whereas the time-dependent $\vect{H}(t)$ acts on finite-dimensional Hilbert space. In this case only a single qumode is required and is especially suitable for non-autonomous ODEs. Alternatively, one can keep $\vect{H}(t)$ acting only on qumodes whereas the clock register is discretised, in the way outlined in Section~\ref{sec:qubitformalism}. This can be suitable for example for non-autonomous PDEs where the discretisation of the PDE can bring problems, but we know that the preparation of $\rho_0$ is simpler in the qubit setting, since it's just a classical mixture of $\{|i\rangle\}_{i=-N/2}^{N/2}$ states. We will examine these various settings for different problems and different platforms in upcoming work.  \\

See Figure~\ref{fig:quantumdynamics} for a schematic diagram of the process in Theorem~\ref{thm:two} to retrieve the approximation to $|y(t)\rangle$ when the clock mode is replaced by qubits.

\section{Applications} \label{sec:applications}
Here we will provide applications to general non-autonomous PDEs including quantum dynamics with time-dependent Hamiltonians, any linear non-autonomous PDE as well as non-autonomous nonlinear scalar hyperbolic and Hamilton-Jacobi PDEs and systems of nonlinear ODEs. \\

We will retain the language of continuous-variable qumodes only for simplicity. However, the formalism can be easily translated into qubit-based systems. For example, the auxiliary clock register can easily be discretised so the single qumode can be replaced by $\log_2 N$ qubits where $N$ is the size of the discretisation,  defined in Section~\ref{sec:qubitformalism}. Either a single qumode or $\log_2 N$ qubits is the maximum amount of extra spatial resources necessary to turn a non-autonomous unitary system into an autonomous one.\\

We can apply the method reviewed in Section~\ref{sec:quantumsimintro} to Schr\"odingerise linear non-autonomous PDEs. The corresponding time-dependent Hamiltonian is then given in Eq.~\eqref{eq:y}. If the time-dependent parameters of this Hamiltonian can be directly manipulated by good quantum control protocols, then no other methods are necessary. However, if a physical architecture is not capable to engineer the time-dependent control of relevant parameters, then we must dilate the system to obtain a time-independent Hamiltonian $\vect{\bar{H}}$, formulated in Theorem~\ref{thm:two} for quantum dynamics and Theorem~\ref{thm:three} for more general linear dynamics. We will write down the corresponding $\vect{\bar{H}}$ for different non-autonomous PDEs. 

\subsection{Quantum dynamics}

We begin with applications to time-dependent Schr\"odinger's equations governing non-autonomous quantum systems. This can also be applied to Schr\"odinger's equations with complex absorbing potentials and even open quantum systems with time-dependent Hamiltonians and system-environment interactions. 

\subsubsection{Closed quantum systems}

There are many problems in quantum physics that involve time-dependent Hamiltonians. Here we will only briefly mention an application relevant for computation is adiabatic quantum computation \cite{albash2018adiabatic}, and leave descriptions of other applications to future work. Here the time dependence in the simplest case has the form
\begin{align*}
    \vect{H}(t)=(1-t)\vect{H}_0+t\vect{H}_f, \qquad t\in [0,1] ,
\end{align*}
where $\vect{H}_0$ and $\vect{H}_f$ are the initial and final Hamiltonians, and the ground state of $\vect{H}_f$ can embed the solution of certain computational problems. Even for such a simple time-dependency, it is clear that the Hamiltonians don't commute at different times
\begin{align*}
[\vect{H}(t), \vect{H}(t')]=(t'-t)[\vect{H}_0, \vect{H}_f] \neq 0, \qquad t\neq t' ,
\end{align*}
since we don't want the ground state of the initial Hamiltonian to coincide with the ground state of the final Hamiltonian. From our formalism, this evolution can be achieved using a time-independent Hamiltonian
\begin{align*}
    \vect{\bar{H}}=\mathbf{1}\otimes \hat{p}_s+\vect{H}(\hat{s})=\mathbf{1}\otimes \hat{p}_s+\vect{H}_0 \otimes \mathbf{1}_s+\vect{H}_f \otimes \hat{s} ,
\end{align*}
where the auxiliary mode can be either presented by a continuous variable or $\log_2 N$ qubits (where in the latter case $\hat{p}_s \rightarrow \hat{P}_s$ and $\hat{s} \rightarrow \hat{S}$), depending on the physical implementation. More details on the relevance of our protocol for adiabatic quantum computing will be in an upcoming paper. 

\subsubsection{Open quantum systems}
We consider time-dependent Schr\"odinger's equations with a non-Hermitian Hamiltonian, as a way to model open quantum systems \cite{figueira_2006,Rotter_2009}
\begin{align} \label{eq:openquantum}
\frac{\dd}{\dd t} \vect{u}(t) = -i \vect{V}(t) \vect{u}(t), \qquad \vect{V}(t)=\vect{A}_1(t)-i\vect{A}_2(t), \qquad \vect{A}_1(t)=\vect{A}^{\dagger}_1(t), \qquad \vect{A}_2(t)=\vect{A}^{\dagger}_2(t),
\end{align}
where the \enquote{effective Hamiltonian} $\vect{V}$ is non-Hermitian. This is known as the effective non-Hermitian Hamiltonian formalism. This formalism can arise from e.g., a non-trace-preserving quantum channel $\mathcal{E}(\rho) = \vect{M} \rho \vect{M}^\dagger$ such that $\vect{M}^\dagger \vect{M} \le I$; one can rewrite this quantum channel by $\mathcal{E}(\rho) = e^{-i \vect{V}} \rho\ e^{i \vect{V}^\dagger}$ where the effective Hamiltonian $\vect{V} = i\log(\vect{M})$ is non-Hermitian in general. The operator $\vect{A}_1$ can be physically interpreted as the Hamiltonian for the original closed quantum system, whereas the term $\vect{A}_2$ can be viewed as a correction term from the system-environment interaction \cite{Rotter_2009}.
The dissipative rate of the system is governed by 
\begin{align*}
\frac{\dd}{\dd t} \norm{\vect{u}(t)}^2 = 2 \bra{\vect{u}(t)} \vect{A}_2(t) \ket{\vect{u}(t)}.
\end{align*}
This means that for applications of Eq.~\eqref{eq:openquantum} to open quantum systems, $\vect{A}_1, \vect{A}_2 \neq 0$. From Eq.~\eqref{eq:sigmat}, we see that the quantum state $|u(t)\rangle$ can be prepared using two ancillary qumodes where the augmented Hamiltonian takes the following form 
\begin{align*}
\vect{\bar{H}}&=I_{\eta}\otimes \mathbf{1}\otimes \hat{p}_s +\hat{\eta}\otimes \vect{A}_2(\hat{s})+I_{\eta}\otimes \vect{A}_1(\hat{s}),\qquad \vect{A}_1(\hat{s})=\frac{1}{2}(\vect{V}(\hat{s})+\vect{V}^{\dagger}(\hat{s})), \qquad \vect{A}_2(s)=\frac{i}{2}(\vect{V}(\hat{s})-\vect{V}^{\dagger}(\hat{s})).
\end{align*}
We note that Eq.~\eqref{eq:openquantum} is also relevant to modelling quantum systems with time-dependent Hamiltonians and with artificial boundary conditions like complex absorbing potentials \cite{ABC2023quantum, Lin-LCU}.

\subsection{General ODEs and PDEs}

\subsubsection{Linear ODEs} \label{sec:linearode}
A system of $J$ first-order linear ODEs can be written in the form 
\begin{align*}
    \frac{d\vect{u}}{dt}=-i\vect{A}(t)\vect{u}, \qquad \vect{u} \in \mathbb{R}^J,
\end{align*}
where here $\vect{A}(t)$ is now a finite-matrix and does not consist of operators acting in infinite-dimensional Hilbert space.  In the hybrid continuous-variable discrete-variable setting where we can assume the clock mode to be represented by a qumode, whereas $\vect{u}$ remains discrete-variable (i.e., represented by qubits), the formulation in Theorem~\ref{thm:three} can be directly applied, so the corresponding time-independent Hamiltonian is 
\begin{align*}
    \vect{\bar{H}}=I_{\eta} \otimes \mathbf{1}\otimes \hat{p}_s+\hat{\eta}\otimes \vect{A}_2(\hat{s})+I_{\eta}\otimes \vect{A}_1(\hat{s}), \qquad \vect{A}_1(\hat{s})=\frac{1}{2}(\vect{A}(\hat{s})+\vect{A}^{\dagger}(\hat{s})), \qquad \vect{A}_2(\hat{s})=\frac{i}{2}(\vect{A}(\hat{s})-\vect{A}^{\dagger}(\hat{s})).
\end{align*}
This can easily extend to the fully discrete setting by discretising $s$ and using the formulation in Section~\ref{sec:qubitformalism}. \\

Higher-order derivative cases can also be easily treated by adding more qubits, following methods in \cite{2023analog}. For instance, we can look at second-order ODES for $\vect{u}$  
\begin{align*}
       & \frac{d^2 \vect{u}}{dt^2}+\vect{\Gamma}(t)\frac{d\vect{u}}{dt}+i\vect{A}(t)\vect{u}=0, \qquad \vect{\Gamma}(t)=\vect{\Gamma}_1(t)-i\vect{\Gamma}_2(t)  \\
       & \vect{\Gamma}_1(t)=(\vect{\Gamma}(t)+\vect{\Gamma}^{\dagger}(t))/2=\vect{\Gamma}_1^{\dagger}(t), \qquad \vect{\Gamma}_2(t)=i(\vect{\Gamma}(t)-\vect{\Gamma}^{\dagger}(t))/2=\vect{\Gamma}_1^{\dagger}(t), 
\end{align*}
We can dilate the system $\vect{u}(t) \rightarrow \vect{y}(t)=\vect{u}(t) \otimes |0\rangle+(d\vect{u}/dt) \otimes |1\rangle$. Then the following equation for $\vect{y}(t)$ can be obtained
    \begin{align} 
    \begin{aligned}
      &   \frac{d \vect{y}}{d t}=\begin{pmatrix}
          d\vect{u}/dt \\
          d^2\vect{u}/dt^2
      \end{pmatrix}=-i\vect{V} \vect{y}, \qquad \vect{V}(t)=\begin{pmatrix}
            0 & iI \\
            \vect{A}(t) & -i\vect{\Gamma}(t)
        \end{pmatrix}, \\
        & \vect{V}(t)=\vect{V}_1(t)-i\vect{V}_2(t), \qquad \vect{A}(t)=\vect{A}_1(t)-i\vect{A}_2(t), \\ 
        & \vect{V}_1(t)=(\vect{V}+\vect{V}^{\dagger})/2=\vect{V}_1^{\dagger}, \vect{V}_2=i(\vect{V}-\vect{V}^{\dagger})/2=\vect{V}_2^{\dagger}, \\
        & \vect{A}_1=(\vect{A}+\vect{A}^{\dagger})/2=\vect{A}_1^{\dagger}, \vect{A}_2=i(\vect{A}-\vect{A}^{\dagger})/2=\vect{A}_2^{\dagger}, \\
        & \vect{V}_1(t)=\frac{1}{2}\begin{pmatrix}
            0 & \vect{A}^{\dagger}+iI \\
            \vect{A}-iI & -i(\vect{\Gamma}-\vect{\Gamma}^{\dagger})
        \end{pmatrix}=\vect{A}_1\otimes \sigma_x+\vect{A}_2 \otimes \sigma_y-\frac{I}{2}\otimes \sigma_y+\vect{\Gamma}_2 \otimes \frac{1}{2}(I-\sigma_z), \\
        & \vect{V}_2(t)=\frac{i}{2}\begin{pmatrix}
            0 & -\vect{A}^{\dagger}+iI \\
            \vect{A}+iI & -i(\vect{\Gamma}+\vect{\Gamma}^{\dagger})
        \end{pmatrix}=-\vect{A}_2 \otimes \sigma_x+\vect{A}_1 \otimes \sigma_y-\frac{I}{2}\otimes \sigma_x+\vect{\Gamma}_1 \otimes \frac{1}{2}(I-\sigma_z). 
        \end{aligned}
    \label{eq:v12equation}
    \end{align}
    Applying Schr\"odingerisation to $d \vect{y}/d t=-i\vect{V}(t) \vect{y}$, we have  the transformation $\vect{y}(t) \rightarrow \tilde{\vect{v}}(t)$ obeying
\begin{align*}
     \frac{d \tilde{\vect{v}}}{d t}+i(\vect{V}_2(t) \otimes \hat{\eta}+\vect{V}_1(t) \otimes I) \tilde{\vect{v}}=\frac{d \tilde{\vect{v}}}{d t}=-i \vect{H}(t) \tilde{\vect{v}}, \qquad \vect{H}(t)=\vect{V}_2(t)\otimes \hat{\eta}+\vect{V}_1(t) \otimes I=\vect{H}^{\dagger}(t).
    \end{align*}
It is straightforward to generalise to higher-order derivatives in $t$. For instance, it is simple to see that if there is at most an $n^{\text{th}}$-order derivative in $t$, then we need $\log_2(n)$ auxiliary qubits. Then the corresponding time-independent Hamiltonian becomes
\begin{align*}
    \vect{\bar{H}}=I_{\eta} \otimes \mathbf{1}\otimes \hat{p}_s+\hat{\eta} \otimes \vect{V}_2(\hat{s})+\mathbf{1} \otimes \vect{V}_1(\hat{s}).
\end{align*}
For a system of non-autonomous nonlinear ODEs, we will treat later in Section~\ref{sec:nonlinearode}.

\subsubsection{Linear homogeneous PDEs with first-order time-derivative} 
Consider a linear homogeneous non-autonomous PDE for $u(x,t)$ with first-order derivative in $t$ (and any order of derivative in $x=(x_1, \cdots, x_D)$) in the following form:
\begin{align} \label{eq:ulinearhomo}
    \frac{\partial u}{\partial t}+\sum_{k=1}^K\sum_{j=1}^D a_{k,j}(t, x_1,...,x_D) \frac{\partial^k u}{\partial x_j^k}+b(t, x_1,...,x_D)u=0,
\end{align}
where the coefficients $a_{2k,j}$ has the sign of $(-1)^{k}$ for stability. \\

To embed this problem into a quantum simulator, without discretising $x$, we need to identify the corresponding $\vect{A}(t)$ acting on infinite-dimensional Hilbert space so Eq.~\eqref{eq:ulinearhomo} can be turned into Eq.~\eqref{eq:nonunitary}, where we can write
\begin{align*}
    \vect{u}(t)=\int dx\, u(t, x)|x\rangle, \qquad x \equiv (x_1, \cdots, x_D). 
\end{align*}
Therefore $|u(t)\rangle$ is a quantum system consisting of $D$ qumodes. We denote the quadrature operators of each qumode as $\hat{x}_j$ and $\hat{p}_j$, $j=1, \cdots, D$ where $[\hat{x}_j,\hat{p}_k]=i\delta_{jk}\mathbf{1}$. If we let $|x_j\rangle$ and $|p_j\rangle$ denote the eigenvectors of $\hat{x}_j$ and $\hat{p}_j$ respectively, then $\langle x|p\rangle=\exp(ix_jp_j)/\sqrt{2\pi}$. The position and momentum eigenstates each form a complete eigenbasis so $\int dx_j |x_j\rangle \langle x_j|=\mathbf{1}=\int dp_j |p_j\rangle \langle p_j|$. Here the spatial degree of freedom can be associated with the position operator $x_j \rightarrow \hat{x}_j$ while the momentum operator $\hat{p}_j$ can be associated with the spatial derivative $\hat{p}_j \leftrightarrow -i\partial/\partial x_j$. See \cite{2023analog} for more details and relevant literature. \\

Then from Eq.~\eqref{eq:ulinearhomo}, one obtains the following equation for $\vect{u}(t)$
\begin{align*}
    \frac{d \vect{u}}{dt}=-i\vect{A}(t, \hat{x}_1,...,\hat{x}_D,\hat{p}_1,...,\hat{p}_D) \vect{u}, \qquad \vect{u}(0)=\vect{u}_0 ,
\end{align*}
where 
\begin{align*}
\vect{A}(t, \hat{x}_1,...,\hat{x}_D, \hat{p}_1,...,\hat{p}_D)  =-\sum_{k=1}^K \sum_{j=1}^D a_{k,j}(t, \hat{x}_1,...,\hat{x}_D)i^{k+1}\hat{p}_j^k-ib(t, \hat{x}_1,...,\hat{x}_D).
\end{align*}
Here $\vect{A}$, acting on $\vect{u}(t)$,  is a linear operator. This is easy to see, since any $\partial^k/\partial x_j^k$ gives a contribution of $(i \hat{p}_j)^k$, and any $x$-dependent coefficient $a(t, x_1,...,x_D)$ gives a contribution $a(t, \hat{x}_1,...,\hat{x}_D)$. For unitary dynamics, this corresponds to when $\vect{A}(t)=\vect{A}^{\dagger}(t)=\vect{H}(t)$ and we can directly apply the formulation in Theorem~\ref{thm:two}. Otherwise we can use Theorem~\ref{thm:three}. Thus in the most general case, the corresponding time-independent Hamiltonian from Theorem~\ref{thm:three} is then
\begin{align*}
    \vect{\bar{H}}=I_{\eta}\otimes \mathbf{1}\otimes \hat{p}_s+\hat{\eta}\otimes \vect{A}_2(\hat{s})+I_{\eta}\otimes \vect{A}_1(\hat{s}).
\end{align*}
For example, we can look at first-order homogeneous time-dependent PDEs for $u(t, x_1,...,x_D)$
\begin{align*}
    \frac{\partial u}{\partial t}+\sum_{j=1}^D a_j(t, x_1,...,x_D) \frac{\partial u}{\partial x_j}+b(t, x_1,...,x_D)u=0, \qquad a_j \in \mathbb{R}, b \in \mathbb{C}. 
\end{align*}
where $a_j \in \mathbb{R}$ to maintain stability of the equation. We can rewrite
\begin{align*}
    & \frac{d \vect{u}}{d t}=-i\vect{A}(t)\vect{u}, \qquad \vect{A}=\sum_{j=1}^D a_j (t, \hat{x}_1,...,\hat{x}_D) \hat{p}_j-ib(t, \hat{x}_1,...,\hat{x}_D), \\
    &\vect{A}(t)=\vect{A}_1(t)-i\vect{A}_2(t), \qquad \vect{A}_1(t)=\frac{1}{2}(\vect{A}(t)+\vect{A}^{\dagger}(t))=\vect{A}_1^{\dagger}(t), \qquad \vect{A}_2(t)=\frac{i}{2}(\vect{A}(t)-\vect{A}^{\dagger}(t))=\vect{A}_2^{\dagger}(t).
\end{align*}
Following Schr\"odingerisation, we have
\begin{align} \label{eq:vtilde}
\begin{aligned}
    & \frac{d \tilde{\vect{v}}}{dt}=-i\vect{H}(t) \tilde{\vect{v}}, \qquad \vect{H}(t)=\vect{A}_2(t) \otimes \hat{\eta}+ \vect{A}_1(t) \otimes I=\vect{H}^{\dagger}, \\
    &  \vect{A}_1(t)=\frac{1}{2}\sum_{j=1}^D\{\hat{p}_j, a_j(t, \hat{x}_1,...,\hat{x}_D)\}+\frac{i}{2}(b^*(t, \hat{x}_1,...,\hat{x}_D)-b(t, \hat{x}_1,...,\hat{x}_D)), \\
    &\vect{A}_2(t)=\frac{i}{2}\sum_{j=1}^D[a_j(t, \hat{x}_1,...,\hat{x}_D),\hat{p}_j]+\frac{1}{2}(b(t, \hat{x}_1,...,\hat{x}_D)+b^*(t, \hat{x}_1,...,\hat{x}_D)).
    \end{aligned}
\end{align}
Thus the corresponding time-independent Hamiltonian becomes
\begin{align} \label{eq:htildefirstorder}
\begin{aligned}
    & \vect{\bar{H}}=I_{\eta}\otimes \mathbf{f} \otimes \hat{p}_s+\hat{\eta} \otimes \left(\frac{i}{2}\sum_{j=1}^D[a_j(\hat{s}, \hat{x}_1,...,\hat{x}_D),\hat{p}_j]+\frac{1}{2}(b(\hat{s}, \hat{x}_1,...,\hat{x}_D)+b^*(\hat{s}, \hat{x}_1,...,\hat{x}_D))\right) \\
    &+I_{\eta} \otimes \left(\frac{1}{2}\sum_{j=1}^D\{\hat{p}_j, a_j(\hat{s}, \hat{x}_1,...,\hat{x}_D)\}+\frac{i}{2}(b^*(\hat{s}, \hat{x}_1,...,\hat{x}_D)-b(\hat{s}, \hat{x}_1,...,\hat{x}_D))\right)
    \end{aligned}
\end{align}
where we made the replacement $t \rightarrow \hat{s}$ in $a(t, \hat{x}_1, \cdots, \hat{x}_D)$ and $b(t, \hat{x}_1, \cdots, \hat{x}_D)$. Note that since $\hat{s}$ in $\vect{\bar{H}}$ in Eq.~\eqref{eq:htildefirstorder} only acts on the clock mode, the operation commutes with $\hat{p}_j$, $j=1,\cdots, D$. Thus if $a_j(t, x_1, \cdots, x_D)$ has terms that only depend on $t$, then those terms vanish in the commutation relation, and they will not contribute to the term in $\vect{\bar{H}}$ containing $\hat{\eta}$. However, if $b(t, x_1, \cdots, x_D)$ has terms that only depend on $t$, then these terms do still show up in the time-independent Hamiltonian containing $\hat{\eta}$. \\

The simplest example is considering the convection equation with constant $a_j$ in space, which varies linearly in time so $a_j=k_jt$, $k_j \in \mathbb{R}$, and $b=0$. The corresponding time-independent Hamiltonian is then
\begin{align*}
    \vect{\bar{H}}=I_{\eta} \otimes \mathbf{1} \otimes \hat{p}_s+I_{\eta} \otimes \sum_{j=1}^D \hat{p}_j \otimes k_j \hat{s} ,
\end{align*}
which only involves Gaussian operations. A slightly more involved example is that of a $D$-dimensional Liouville equation where $a_j=k_j x_j t$ and $b=lt$ where $k_j, l \in \mathbb{R}$. In this case, $\vect{\bar{H}}$ reduces to 
\begin{align*}
    \vect{\bar{H}}=I_{\eta} \otimes \mathbf{1} \otimes \hat{p}_s+ \hat{\eta} \otimes \frac{1}{2}\left(-\sum_{j=1}^D \mathbf{1} \otimes k_j \hat{s} + \mathbf{1} \otimes l\hat{s}\right) +I_{\eta} \otimes \frac{1}{2}\sum_{j=1}^D (\hat{x}_j\hat{p}_j+\hat{p}_j\hat{x}_j) \otimes k_j \hat{s},
\end{align*}
where the time-independent Hamiltonian has a non-Gaussian term.  \\

We can also look at other examples. For instance, consider the $D$-dimensional general heat equation with a source term 
\begin{align*}
    \frac{\partial u}{\partial t}-\sum_{i=1}^D\frac{\partial}{\partial x_i}\left(\sum_{j=1}^D D_{ij}(t, x_1,...,x_D)\frac{\partial u}{\partial x_j}\right)+V(t, x_1,...,x_D)u=0, \qquad D_{ij}(t, x_1,...,x_D)>0, \, V(t, x_1,...,x_D) \in \mathbb{R},
\end{align*}
where $D_{ij}$ are the components of a symmetric positive definite  diffusion matrix. In this case the corresponding time-independent Hamiltonian is 
\begin{align*}
    \vect{\bar{H}}=I_{\eta} \otimes \mathbf{1} \otimes \hat{p}_s+\hat{\eta} \otimes \vect{A}_2(\hat{s}), \qquad \vect{A}_2=\sum_{i,j=1}^D \hat{p}_iD_{ij}(\hat{s}, \hat{x}_1,...,\hat{x}_D)\hat{p}_j+V(\hat{s},\hat{x}_1,...,\hat{x}_D). 
\end{align*}
We can also consider the Fokker-Planck equation, which gives the time evolution of the probability density function $u(t,x_1,...,x_D)$ of the velocity of a particle under the impact of drag and random forces. In the Fokker-Planck equation below, $\mu_j$ are the components of the drift vector and $D_{j}$ are the components of the diffusion vector
\begin{align*}
\frac{\partial u}{\partial t}+\sum_{j=1}^D \frac{\partial}{\partial x_j}(\mu_j (t, x_1,...,x_D) u)-\sum_{j=1}^D \frac{\partial^2}{\partial x_j^2}(D_{j}(t, x_1,...,x_D) u)
=0, \qquad \mu_j, D_{j} \in \mathbb{R}.
\end{align*}
The drift vector and diffusion coefficients $D_{j}$ only have real-valued components, and $D_j>0$. The corresponding time-independent Hamiltonian then becomes
\begin{align*}
     \vect{\bar{H}}=I_{\eta} \otimes \mathbf{1} \otimes \hat{p}_s+\hat{\eta} \otimes \vect{A}_2(\hat{s})+I_{\eta} \otimes \vect{A}_1(\hat{s}) ,
\end{align*}
with
\begin{align*}
    & \vect{A}_1(\hat{s})=\frac{1}{2}\sum_{j=1}^D\{\hat{p}_j,\mu_j(\hat{s}, \hat{x}_1,...,\hat{x}_D)\}-\frac{i}{2}\sum_{j=1}^D[\hat{p}^2_j, D_{j}(\hat{s}, \hat{x}_1,...,\hat{x}_D)]  \\
    & \vect{A}_2(\hat{s})=\frac{i}{2}\sum_{j=1}^D [\mu_j(\hat{s}, \hat{x}_1,...,\hat{x}_D),\hat{p}_j]-\frac{1}{2}\sum_{j=1}^D\{\hat{p}_j^2,  D_{j}(\hat{s}, \hat{x}_1,...,\hat{x}_D)\} .
\end{align*} 
Suppose we have additive noise $D_j=a_j$ and when the drift is linear in time 
$$
\mu_j=c_jt, \qquad  \vect{A}=-\sum_{j=1}^D c_j \hat{p}_j \hat{s}-i\sum_{j=1}^D a_j\hat{p}_j^2 ,
$$
and the corresponding time-independent Hamiltonian in $D+2$ qumodes is 
\begin{align*}
    \vect{\bar{H}}=I_{\eta} \otimes \mathbf{1} \otimes \hat{p}_s+\hat{\eta} \otimes \frac{1}{2}\sum_{j=1}^D \left(c_j I-a_j\hat{p}_j^2\right)\otimes I_s+I_{\eta} \otimes \sum_{j=1}^D c_j\hat{p}_j\otimes \hat{s}.
\end{align*}
\subsubsection{Linear inhomogeneous PDEs with first-order time-derivative}
A linear inhomogeneous PDE for $u(x,t)$ with first derivative in $t$  with an inhomogeneous term $f \neq 0$ can be written as
\begin{align} \label{eq:ulinear2}
    \frac{\partial u}{\partial t}+\sum_{k=1}^K\sum_{j=1}^D a_{k,j}(t, x_1,...,x_D) \frac{\partial^k u}{\partial x_j^k}+b(t, x_1,...,x_D)u=f(t, x_1,...,x_D),
\end{align}
where the coefficients satisfy  $\text{sign}(a_{2k,j})=(-1)^k$. We can rewrite this as 
\begin{align*}
    \frac{d \vect{u}}{dt}=-i\vect{A}(t, \hat{x}_1,...\hat{x}_D, \hat{p}_1,...,\hat{p}_D)\vect{u}+\vect{f}(t, \hat{x}_1,...,\hat{x}_D).
\end{align*}
In the case where the time-dependence of the inhomogeneous term can be factored out like $f(t, x_1, \cdots, x_D)=\exp(g_1(t))g_2(x_1, \cdots, x_D)$, then we only need to augment the system by a single qubit and perform the dilation $\vect{y} \rightarrow \vect{u} \otimes |0\rangle+\vect{f} \otimes |1\rangle$ and we have
\begin{align}
\begin{aligned}
  &  \frac{d \vect{y}}{d t}= \frac{d}{d t}\begin{pmatrix}
        \vect{u} \\
        \vect{f}
    \end{pmatrix}=-i\vect{B}(t) \vect{y}, \qquad \vect{B}(t)=\begin{pmatrix}
        \vect{A}(t)& iI \\
        0 & idg_1(t)/dt
    \end{pmatrix}, \\
    & \vect{A}(t)=\vect{A}_1(t)-i\vect{A}_2(t),  \qquad \vect{B}(t)=\vect{B}_1(t)-i\vect{B}_2(t),  \\
    & \vect{B}_1(t)=(\vect{B}(t)+\vect{B}^{\dagger}(t))/2=\vect{B}_1^{\dagger}(t), \qquad \vect{B}_2(t)=i(\vect{B}(t)-\vect{B}^{\dagger}(t))/2=\vect{B}_2^{\dagger}(t), \\
    & \vect{B}_1(t)=\begin{pmatrix}
        \vect{A}_1(t) & iI/2 \\
        -iI/2 & \frac{i}{2}\frac{d(g_1-g_1^*)}{dt}
    \end{pmatrix}=\vect{A}_1(t) \otimes \frac{1}{2}(I+\sigma_z)+\frac{I}{2} \otimes \sigma_y+ \frac{i}{2}\frac{d(g_1-g_1^*)}{dt}\mathbf{1}\otimes \frac{1}{2}(I-\sigma_z), \\
    & \vect{B}_2(t)=\begin{pmatrix}
        \vect{A}_2(t) &  -I/2 \\
         -I/2 & \frac{1}{2}\frac{d(g_1+g_1^*)}{dt}
    \end{pmatrix}=\vect{A}_2(t) \otimes \frac{1}{2}(I+\sigma_z)-\frac{I}{2} \otimes \sigma_x+ \frac{1}{2}\frac{d(g_1+g_1^*)}{dt}\mathbf{1}\otimes \frac{1}{2}(I-\sigma_z), 
    \end{aligned}
\end{align}
where $I=\begin{pmatrix}
    1 & 0 \\
    0 & 1
\end{pmatrix}$. Note that the functions $dg_1/dt$ and $dg_1^*/dt$ can be computed classically offline. Then the time-dependent Hamiltonian after Schr\"odingerisation acting on $D+1$ qumodes and a single qubit is 
\begin{align*}
    \vect{H}(t)=\hat{\eta} \otimes \vect{B}_2(t)+I_{\eta} \otimes \vect{B}_1(t).
    \end{align*}
The corresponding time-independent Hamiltonian from Theorem~\ref{thm:three} is then
\begin{align*}
    \vect{\bar{H}}=I_{\eta} \otimes \mathbf{1}\otimes \hat{p}_s+\hat{\eta} \otimes \vect{B}_2(\hat{s})+\mathbf{1} \otimes \vect{B}_1(\hat{s}), 
\end{align*}
which is a Hamiltonian acting on a system of $D+2$ qumodes and a single qubit. 

\subsubsection{Linear homogeneous PDEs with second-order time derivatives and beyond}
It is also straightforward to extend to higher-order derivatives in time, following the same methods in \cite{2023analog}. We look at the following homogeneous linear PDE (includes for example the wave equation as a special case)
\begin{align} \label{eq:t2}
    \frac{\partial^2 u}{\partial t^2}+c_0(t, x_1,...,x_D) \frac{\partial u}{\partial t}+\sum_{j=1}^d c_j(t, x_1,...,x_D) \frac{\partial^2 u}{\partial x_j \partial t}+\sum_{l=1}^L\sum_{j=1}^D a_{j,l}(t, x_1,...,x_D) \frac{\partial^l u}{\partial x_j^l}+b(t, x_1,...,x_D)u=0,
\end{align}
where the coefficients $a_{j,l}<0$ when $l$ is an even integer.
One can rewrite Eq.~\eqref{eq:t2} into the following equation for $\vect{u}$
\begin{align*}
       & \frac{d^2 \vect{u}}{dt^2}+\vect{\Gamma}(t)\frac{d\vect{u}}{dt}+i\vect{A}(t)\vect{u}=0, \qquad \vect{\Gamma}(t) \equiv \vect{c}_0(t)+i\sum_{j=1}^D \vect{c}_j(t)(I^{\otimes j-1}\otimes \hat{p}_j\otimes I^{\otimes D-j}), \qquad \vect{\Gamma}(t)=\vect{\Gamma}_1(t)-i\vect{\Gamma}_2(t) , \\
       & \vect{\Gamma}_1(t)=(\vect{\Gamma}(t)+\vect{\Gamma}^{\dagger}(t))/2=\vect{\Gamma}_1^{\dagger}(t), \vect{\Gamma}_2(t)=i(\vect{\Gamma}(t)-\vect{\Gamma}^{\dagger}(t))/2=\vect{\Gamma}_1^{\dagger}(t),  \\
       & \vect{A}(t)=-i\left(\sum_{l=1}^L\sum_{j=1}^D a_{j,l}(t, \hat{x_1},...,\hat{x}_D)(i\hat{p}_j)^l+b(t, \hat{x}_1,...,\hat{x}_D)\right)
\end{align*}
where $\vect{c}_i(t)=c_i(t, \hat{x_1},...,\hat{x}_D)$.  From Section~\ref{sec:linearode}, we can dilate the system $\vect{u}(t) \rightarrow \vect{y}(t)=\vect{u}(t) \otimes |0\rangle+(d\vect{u}/dt) \otimes |1\rangle$. The formulation with Section~\ref{sec:linearode} is the same, except here $\vect{\Gamma}(t)$ and $\vect{A}(t)$ act on infinite dimensional Hilbert space. Then the corresponding time-independent Hamiltonian becomes
\begin{align*}
    \vect{\bar{H}}=I_{\eta} \otimes \mathbf{1}\otimes \hat{p}_s+\hat{\eta} \otimes \vect{V}_2(\hat{s})+\mathbf{1} \otimes \vect{V}_1(\hat{s}), 
\end{align*}
with $\vect{V}_1(\hat{s}), \vect{V}_1(\hat{s})$ given in Eq.~\eqref{eq:v12equation}. This is easily generalisable to higher-order derivatives in $t$. If there is at most an $n^{\text{th}}$-order derivative in $t$, then we need $\log_2(n)$ auxiliary qubits.\\

For instance, consider a $D$-dimensional wave equation of the form
\begin{align*}
    \frac{\partial^2 u}{\partial t^2}-\sum_{j=1}^D a_j(t, x_1,...,x_D) \frac{\partial^2 u}{\partial x_j^2}=-V(t, x_1,...,x_D)u, \qquad u(0, x), \qquad x_j \in \mathbb{R}^D, a_j \in \mathbb{R}^+, V \in \mathbb{R}. 
\end{align*}
In this case
\begin{align*}
& \vect{V}_1(\hat{s})=\vect{A}_1 (\hat{s}) \otimes \sigma_x+\left(\vect{A}_2(\hat{s})-\frac{I}{2}\right)\otimes \sigma_y, \qquad \vect{V}_2(t)=-\vect{A}_1(\hat{s}) \otimes \sigma_y+\left(\vect{A}_2(\hat{s})+\frac{I}{2}\right)\otimes \sigma_x,  \\
& \vect{A}_1(\hat{s})=\frac{i}{2}\sum_{j=1}^D [\hat{p}^2_j,a_j(\hat{s}, \hat{x}_1,...,\hat{x}_D)], \qquad \vect{A}_2(\hat{s})=\frac{1}{2}\sum_{j=1}^D \{\hat{p}_j^2,a_j(\hat{s}, \hat{x}_1,...,\hat{x}_D)\}+V(\hat{s}, \hat{x}_1,...,\hat{x}_D).
\end{align*}

\subsubsection{Nonlinear ODEs} \label{sec:nonlinearode}
Suppose we have a system of $J$ nonlinear ODEs for $\gamma_n(t)$ where each $\gamma_n(t)$ obeys a nonlinear ODE
\begin{align*}
    \frac{d\gamma_n(t)}{dt}=F_n(t, \gamma_0,...,\gamma_{J-1}), \qquad n=0,...,J-1 , 
\end{align*}
and $F_n$ is a nonlinear function of its arguments. Here we introduce $J$ auxiliary variables $q_0,...,q_{J-1}$ and a function $\Phi(t, q_0,...,q_{J-1})$ defined by 
\begin{align*}
    \Phi(t,q_0,...,q_{J-1})=\prod_{n=0}^{J-1}\delta(q_n-\gamma_n(t)), \qquad q_n \in \mathbb{R}.
\end{align*}
Then it is simple to check that $\Phi(t,q_0,...,q_{N-1})$ satisfies, in the weak sense \cite{JinOsher, jin2022quantum}, the \textit{linear} $N+1$-dimensional PDE
\begin{align} \label{eq:nonlinearphi}
    \frac{\partial \Phi(t,q_0,...,q_{J-1})}{\partial t}+\sum_{n=0}^{J-1}\frac{\partial}{\partial q_n}(F_n(t, q_0,...,q_{J-1})\Phi(t,q_0,...,q_{N-1}))=0.
\end{align}
We can treat this linear PDE with the same methods as before. We can define $\hat{q}_n$ to be a quadrature operator whose eigenstate is $|q_n\rangle$ and its conjugate quadrature operator as $\hat{Q}_n$. Thus $[\hat{q}_n, \hat{Q}_n]=i$ and we can use the correspondence $-i\hat{Q}_n \leftrightarrow \partial/\partial q_n$. Then defining  $\vect{\Phi}(t) \equiv \int \Phi(t,q_0,...,q_{N-1})|q_0,...,q_{N-1}\rangle dq_0...dq_{N-1}$, we see that Eq.~\eqref{eq:nonlinearphi} becomes
\begin{align*}
    \frac{d\vect{\Phi}(t)}{dt}=-i\vect{A}(t)\vect{\Phi}(t), \qquad \vect{A}(t)=\sum_{n=0}^{N-1}\hat{Q}_n F_n(t, \hat{q}_0,...,\hat{q}_{N-1}).
\end{align*}
Using Schr\"odingerisation $\vect{\Phi}(t) \rightarrow \tilde{\vect{v}}(t)$ we obtain
\begin{align*}
   &  \frac{d \tilde{\vect{v}}}{dt}=-i\vect{H}(t)\tilde{\vect{v}}, \qquad \vect{H}(t)=\vect{A}_2(t)\otimes \hat{\eta}+\vect{A}_1(t)\otimes I=\vect{H}^{\dagger}(t) , \\
   & \vect{A}_1(t)=\frac{1}{2}\sum_{n=0}^{J-1}\{\hat{Q}_n, F_n(t, \hat{q}_0,...,\hat{q}_{J-1})\}, \qquad \vect{A}_2(t)=\frac{i}{2}\sum_{n=0}^{J-1}[\hat{Q}_n, F_n(t, \hat{q}_0,...,\hat{q}_{J-1})] ,
\end{align*}
where $\vect{H}(t)$ is a time-dependent Hamiltonian on $J+1$ qumodes. From Theorem~\ref{thm:three}, the corresponding time-independent Hamiltonian on $J+2$ qumodes is then 
\begin{align*}
    \vect{\bar{H}}=I_{\eta}\otimes \mathbf{1}\otimes \hat{p}_s+\hat{\eta}\otimes \vect{A}_2(\hat{s})+I_{\eta}\otimes \vect{A}_1(\hat{s}).
\end{align*}
This method can also be applied to an approximation of general nonlinear PDEs, which will first be discretised spatially to form a system of ODEs.  Here $J$ can be larger or smaller depending on the discretisation technique. \\

\subsubsection{Nonlinear PDEs}

There are some nonlinear PDEs where one does not need to take spatial discretisation first to form a system of nonlinear ODEs. Such problems belong to the cases where we can find a linear representation in higher dimension which is linear and equivalent to the original PDEs.  These include the scalar hyperbolic and Hamilton-Jacobi equations \cite{2022nonlinearpde, 2023nonlinearjcp}. \\

{\it Scalar nonlinear hyperbolic PDEs}

A $D$-dimensional scalar nonlinear hyperbolic PDEs for $u(t, x_1, \cdots, x_D)$ with explicit time-dependence can be written
\begin{align}\label{hyp-PDE2}
  \frac{\partial u}{\partial t} + \sum_{j=1}^D F_j(t, x_1, \cdots, x_D, u)\frac{\partial u}{\partial x_j} + Q(t, x_1,...,x_D,u)=0, \qquad  u(0,x)=u_0(x) ,
  \end{align}
  Following the same formalism as in \cite{2022nonlinearpde}, we introduce a level set function $\phi(t,x_1,...,x_D,\chi)$, where $\chi \in \mathbb{R}$. It is defined so its zero level set is the solution
  $u$:
  \begin{equation*}
        \phi(t,x_1,...,x_D,\chi)=0  \quad {\text{at}} \quad \chi=u(t,x_1,...,x_D) .
  \end{equation*}
  By defining
  \begin{align*}
    \Psi(t, x_1,...,x_D, \chi)=\delta(\phi(t,x_1,...,x_D,\chi)),
\end{align*}
it can shown that $\Psi$ (like $\phi$) evolves according to the linear PDE in one extra dimension
\begin{align*}
     & \frac{\partial \Psi}{\partial t} + \sum_{j=1}^{D} F_j(t, x_1, \cdots, x_D, \chi)\frac{\partial \Psi}{\partial x_j} +Q(t, x_1, \cdots, x_D, \chi) \frac{\partial \Psi}{\partial \chi}= 0\\
      & \Psi(0,x_1,...,x_D,\chi)=\delta(\chi-u_0(x_1,...,x_D)).
\end{align*}
If $\vect{\Psi}(t) \equiv \int dx_1 \cdots dx_D d \chi \Psi(t, x_1, \cdots, x_D, \chi)|x_1, \cdots, x_D, \chi \rangle$, then it satisfies the following linear ODE 
\begin{align*}
\frac{d \vect{\Psi}}{dt}=-i\vect{A}(t)\vect{\Psi}, \qquad \vect{A}(t)=\sum_{j=1}^D F_j (t, \hat{x}_1, \cdots, \hat{x}_D, \hat{\chi}) \hat{p}_j+Q(t, \hat{x}_1,...,\hat{x}_D, \hat{\chi}) \hat{\zeta} ,
\end{align*}
where $i\hat{p}_j=\partial/\partial x_j$, $i\hat{\zeta}=\partial/\partial \chi$. We can apply Schr\"odingerisation to this system $\vect{\Psi} \rightarrow \tilde{\vect{v}}$ to get 
\begin{align*}
   & \frac{d\tilde{\vect{v}}}{dt}=-i\vect{H}(t)\tilde{\vect{v}}, \qquad \vect{H}(t)=\vect{A}_2(t) \otimes \hat{\eta}+\vect{A}_1(t) \otimes I=\vect{H}^{\dagger} \\
   & \vect{A}_1(t)=\sum_{j=1}^DF_j(t, \hat{x}_1, \cdots, \hat{x}_D, \hat{\chi})\hat{p}_j+\frac{1}{2}\sum_{j=1}^D\{\hat{\zeta}, Q(t, \hat{x}_1,...,\hat{x}_D, \hat{\chi})\}, \qquad \vect{A}_2(t)=\frac{i}{2}\sum_{j=1}^D[ \hat{\zeta}, Q(t, \hat{x}_1,...,\hat{x}_D, \hat{\chi})].
\end{align*} 
Here $\vect{\Psi}$ consists of $D+1$ modes and the Hamiltonian $\vect{H}$ operates on $D+2$ modes. Then from Theorem~\ref{thm:three}, the corresponding time-independent Hamiltonian acting on $D+3$ qumodes is then
\begin{align*}
\begin{aligned}
    \vect{\bar{H}}&=I_{\eta} \otimes \mathbf{1} \otimes \hat{p}_s+\hat{\eta} \otimes \frac{i}{2}\sum_{j=1}^D[\hat{\zeta}, Q(\hat{s}, \hat{x}_1,...,\hat{x}_D, \hat{\chi})] \\
    &\qquad +I_{\eta} \otimes \left(\sum_{j=1}^D F_j(\hat{s}, \hat{x}_1, \cdots, \hat{x}_D, \hat{\chi})\hat{p}_j+\frac{1}{2}\sum_{j=1}^D\{\hat{\zeta}, Q(\hat{s}, \hat{x}_1,...,\hat{x}_D, \hat{\chi})\}\right). 
    \end{aligned}
\end{align*}

{\it Hamilton-Jacobi equations}

A similar methodology applies to nonlinear Hamilton-Jacobi equations. If $S$ obeys the nonlinear Hamilton-Jacobi PDE and $u=\grad S$, then $u$ solves the nonlinear hyperbolic system of conservation equations in gradient form
\begin{align*}
    \frac{\partial u_j}{\partial t}+\frac{\partial}{\partial x_j} H(t, x_1,...,x_D, u_1,...,u_D)=0, \qquad \forall j=1,...,D, \qquad H \in \mathbb{R}
\end{align*}
where the function $H=H(t, x_1,...,x_D, p_1,...,p_2)$ is a Hamiltonian. We now introduce $D$ variables $\chi_1,...,\chi_D \in \mathbb{R}$. 
Then one can define a system of level set functions $\phi=(\phi_1,...,\phi_D)$ where each $\phi_j(t,x_1,...,x_D, \chi_1,...,\chi_D)=0$ at $\chi_j=u_j$. We can then define 
\begin{align*}
    \Psi(t,x_1,...,x_D, \chi_1,...,\chi_D)=\prod_{j=1}^D \delta(\phi_j(t,x_1,...,x_D, \chi_1,...,\chi_D)), \qquad \Psi(t=0)=\prod_{j=1}^D \delta(\chi_j-u_j(t=0)).
\end{align*}
which obeys the following linear transport PDE
\begin{align*}
     \frac{\partial \Psi}{\partial t}+\sum_{j=1}^D \frac{\partial H}{\partial \chi_j}\frac{\partial \Psi}{\partial x_j}-\sum_{j=1}^D\frac{\partial H}{\partial x_j}\frac{\partial \Psi}{\partial \chi_j}=0.
\end{align*}
Then 
\begin{align*}
    \frac{d \vect{\Psi}}{dt}=-i\vect{H}(t)\vect{\Psi}, \qquad \vect{H}(t)=i\sum_{j=1}^D(\hat{\zeta}_jH(t, \hat{x}_1,...,\hat{x}_D, \hat{p}_1,...,\hat{p}_D) \hat{p}_j-\hat{p}_jH(t, \hat{x}_1,...,\hat{x}_D, \hat{p}_1,...,\hat{p}_D)\hat{\zeta}_j)=\vect{H}^{\dagger}(t).
\end{align*}
which acts on $2D$ qumodes. Since $\vect{H}(t)=\vect{H}^{\dagger}(t)$ because $H$ is interpreted as a Hamiltonian and each of the quadratures is hermitian, no Schr\"odingerisation is required and we can directly apply Theorem~\ref{thm:three} to obtain the corresponding time-independent Hamiltonian on $2D+1$ qumodes
\begin{align*}
    \vect{\bar{H}}=\mathbf{1}\otimes \hat{p}_s+\vect{H}(\hat{s}), \qquad \vect{H}(\hat{s})=i\sum_{j=1}^D(\hat{\zeta}_jH(\hat{s}, \hat{x}_1,...,\hat{x}_D, \hat{p}_1,...,\hat{p}_D) \hat{p}_j-\hat{p}_jH(\hat{s}, \hat{x}_1,...,\hat{x}_D, \hat{p}_1,...,\hat{p}_D)\hat{\zeta}_j).
\end{align*}
For both the scalar hyperbolic and Hamilton-Jacobi equations, observables of the PDE can be recovered from $\vect{\Psi}$ using techniques in \cite{2022nonlinearpde}.

We point out that the level set methods for the above quasi-linear PDEs produce {\it multi-valued} solutions \cite{JinOsher}, not the viscosity solutions \cite{CL83, Lax-notes}. The multi-valued solution corresponds to solutions obtained by the method of characteristics,
and has applications in semi-classical limit of quantum mechanics, geometric optics, and seismic waves (for multiple arrivals), and high-frequency limits of general linear waves, for examples \cite{JinOsher}.

 \section{Numerical experiments} \label{sec:numerical}
 
We use a classical computer to emulate the above quantum algorithms for a few time-dependent dynamics with easily accessible solutions for validation: a closed Hamiltonian dynamics, a dynamics for open quantum system characterised via a non-Hermitian linear operator, as well as a 1D Fokker-Planck equation. Emulating high-dimensional PDEs and quantifying the actual quantum resources necessary in practice are interesting, but are beyond the scope of this current work, and we shall leave this for future investigations.\\
 
\subsection{Simulation details}
 
 \smallskip
 To simulate Eqs.~\eqref{w-eqn} and \eqref{eq:sigmat}, we use Galerkin's method to project the $s$-qumode to $N_s$ basis functions with parameter $\varsigma_s$ (which amounts to using $\log_2(N_s)$ additional qubits); similarly, for the $\eta$-qumode for Schr{\"o}dingerisation, we project the wave function into $N_\eta$ basis functions with parameter $\varsigma_\eta$. If the original system is discrete with dimension $N_u$, there is no need to perform any approximation; in case, the original system is itself a PDE (for e.g., the Fokker-Planck equation below), we similarly project this part into $N_u$ basis functions with parameter $\varsigma_u$. Please refer to Appendix~\ref{app::numerics} for details about the family of basis functions used as well as the meaning of $\varsigma_s$ (as well as $\varsigma_\eta$, $\varsigma_u$).  We choose the projection operator $\hat{P}_{\text{imp}} = \int_{0}^{2}d\xi\ \ket{\xi}\bra{\xi}$ (see Theorem~\ref{thm:three}); similarly, we approximate this projection operator using the basis for $\eta$-qumode.

\subsection{A Hamiltonian system}

We consider a simple example for numerical validation. The dimension $n = 2$ and the time-dependent Hamiltonian is $H(s) = a g(s) \vect{h}$ where $\vect{h} = \frac{1}{2}\sigma_X + \frac{1}{3} \sigma_Y + \frac{1}{4} \sigma_Z$, $a > 0$ is a parameter, $g(s)$ is a polynomial of $s$, and $\sigma_X, \sigma_Y$ and $\sigma_Z$ are Pauli matrices.  We choose the initial state to be $\vect{y}_0 = \begin{bsmallmatrix}\frac{1}{\sqrt{2}}\\ \frac{1}{\sqrt{2}}\end{bsmallmatrix}$. The exact solution is accessible for easy validation. 
To simulate \eqref{w-eqn}, we use $N_s = 32$ and $\varsigma_s=0.2$ (which amounts to using $\log_2(32)=5$ additional qubits, and please refer to Appendix~\ref{app::numerics} for details).  In Figure~\ref{fig::hami}, the discrete data points are \emph{errors using the quantum algorithm} in Theorem~\ref{thm:three}, whereas dashed lines represent the \emph{theoretical prediction} of the error by Lemma~\ref{lemma::error_quantum}. When $\omega$ is small enough, clearly, the theoretical estimate agrees with the actual error of $O(\omega^2)$ in general.
 This example not only demonstrate the validity of the first protocol of quantum algorithms by discarding the clock mode ($s$-mode), as well as the asymptotic error scaling.
 
\begin{figure}
\begin{subfigure}[b]{0.31\textwidth}
\includegraphics[width=\textwidth]{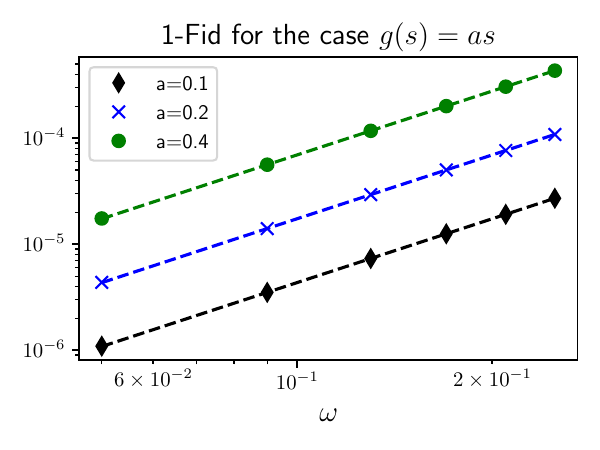}
\caption{Linear function with $T=\nicefrac{1}{2}$}
\end{subfigure}
~
\begin{subfigure}[b]{0.31\textwidth}
\includegraphics[width=\textwidth]{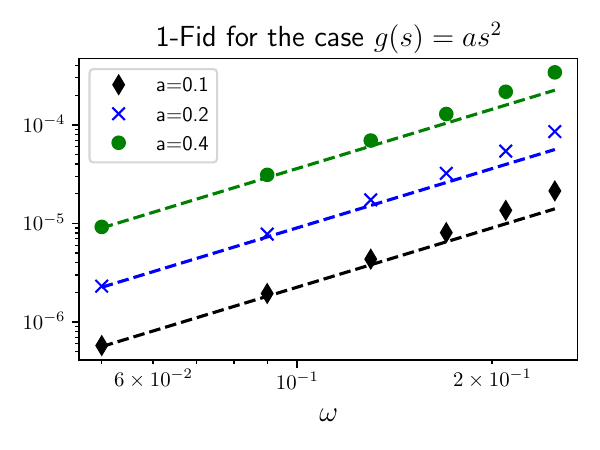}
\caption{Quadratic function with $T=\nicefrac{3}{5}$}
\end{subfigure}
~
\begin{subfigure}[b]{0.31\textwidth}
\includegraphics[width=\textwidth]{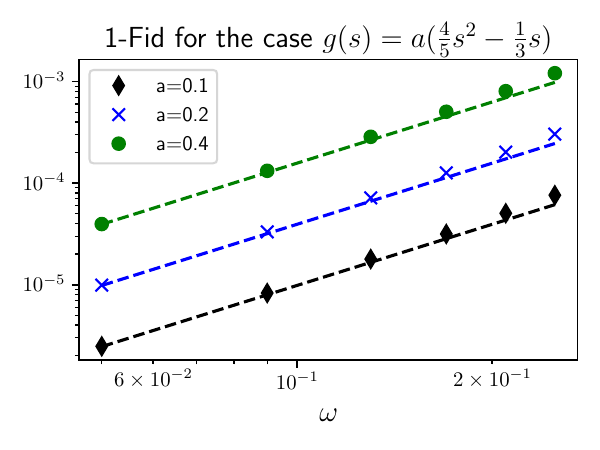}
\caption{Quadratic function with $T=\nicefrac{6}{5}$}
\end{subfigure}
\caption{The error quantified by $1-\text{Fid}$ with respect to $\omega$ for various time functions $g$ and magnitude $a$. The discrete data points are \emph{errors using the quantum algorithm} in Theorem~\ref{thm:three}, whereas dashed lines represent the \emph{theoretical prediction} of the error by Lemma~\ref{lemma::error_quantum}, namely, we draw the line $\mathsf{C}\omega^2$.}
\label{fig::hami}
\end{figure}

\subsection{An open quantum system (a general linear ODE)}

 We consider a two-dimensional ODE with 
 \begin{align}
 \label{eg::ode}
 A_1(s) = a g(s) \begin{bmatrix} \frac{3}{5} & 0 \\ 0 & \frac{7}{5} \end{bmatrix},\qquad A_2(s) = a g(s) \begin{bmatrix} \frac{5}{4} & i \\ -i & \frac{5}{4} \end{bmatrix},\qquad u_0 = \begin{bmatrix}\frac{\sqrt{2}}{\sqrt{3}} \\ \frac{1}{\sqrt{3}}\end{bmatrix},
 \end{align}
 where $a$ characterises the magnitude and $g$ characterises the time-dependence. 
We use an estimator $\widehat{O} = \sigma_Z$ to observe the accuracy that the quantum algorithm achieves. We used $N_s = 64$, $\varsigma_s=0.2$, $N_\eta = 64$, $\varsigma_\eta=2.0$ (see Appendix~\ref{app::numerics}), which amounts to using $6$ qubits for time dilation and $6$ qubits for Schr{\"o}dingerisation. 
 As is clear in Figure~\ref{fig::ode}, with accurate enough $\omega$, the quantum algorithms produces reasonably accurate approximation of the exact value as long as $\omega$ is sufficiently small.
 
 \begin{figure}
 \begin{subfigure}[b]{0.45\textwidth}
 \includegraphics[width=\textwidth]{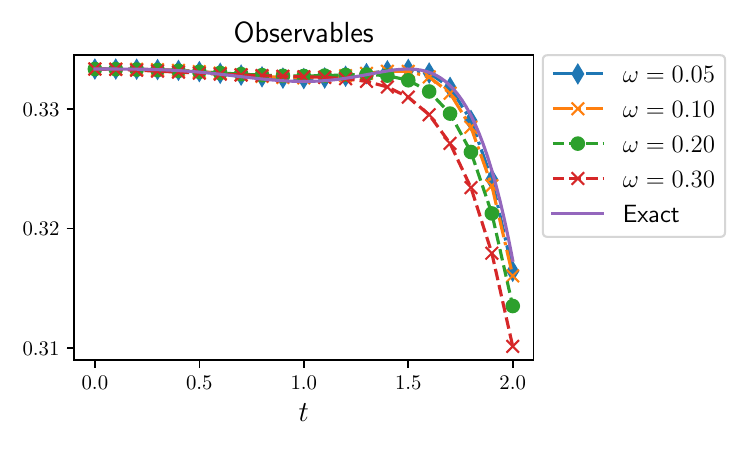}
 \caption{$a=0.3$, $g(s) = s^2 - s$}
 \end{subfigure}
 ~
  \begin{subfigure}[b]{0.45\textwidth}
 \includegraphics[width=\textwidth]{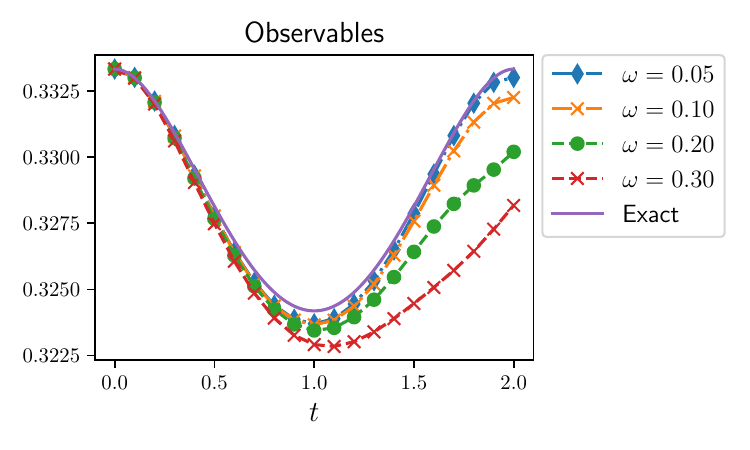}
  \caption{$a=0.3$, $g(s) = 1-s$}
 \end{subfigure}
 \caption{We consider a two-dimensional ODE in \eqref{eg::ode} and visualise observables $\expval{\sigma_Z}$ along the time for the normalised state $\ket{\vect{u(t)}}$ for various time-dependence $g$ and $\omega$. 
 }
 \label{fig::ode}
 \end{figure}

\subsection{1D Fokker-Planck equation.}

We consider the time-dependent Fokker-Planck equation
\begin{align*}
\partial_t \vect{q}(t,x) = g(t) \div \big(x \vect{q}(t,x)\big) + \beta(t) \Delta \vect{q}(t,x) =: -i \vect{A}  \vect{q}(t,x),
\end{align*}
which characterises the evolution of the probability density function for a time-dependent Ornstein-Uhlenbeck process
$dX(t)= - \partial_x U\big(t, X(t)\big)\ dt + \sqrt{2 \beta(t)}\ d W(t)$, where the time-dependent potential $U(t,x) = g(t) \frac{x^2}{2}$.
By Theorem~\ref{thm:three}, the Hamiltonian conservation part $\vect{A}_1$ and the interaction part $\vect{A}_2$ are 
\begin{align*}
\vect{A}_1 = - g(t) \hat{x}\ \hat{p} + i \frac{g(t)}{2}, \qquad \vect{A}_2 = - \frac{g(t)}{2} + \beta(t) \hat{p}^2.
\end{align*}

We use the following observables $\expval{\hat{x}}{q(t)}$, $\expval{\hat{x}^2}{q(t)}$ where $\ket{q(t)}$ is the normalised quantum state of $\vect{q}(t,x)$. In total, we consider three cases where explicit solutions are available: 
\begin{align}
\label{eqn::fp_eg}
\left\{
\begin{aligned}
\text{Case (1):}&\qquad  g(t) = 0.5 s,\ \beta(t) = 0.5 s;\\
\text{Case (2):}&\qquad g(t) = 0.5 s,\ \beta(t) = 0.3;\\
\text{Case (3):}&\qquad g(t) = 0.5 s^3, \beta(t) = 0.3 s.
\end{aligned}\right.
\end{align}

In all three cases, the initial condition $\vect{q}(0,\cdot) = \mathcal{N}\big(0.8, 0.3^2\big)$. In simulation, we used the Galerkin method with parameters $N_s=128$, $\varsigma_s=0.2$,  $N_\eta=128$, $\varsigma_\eta=2.0$, $N_u=64$, $\varsigma_u=2.0$, which amounts to using $7$ qubits for time dilation, $7$ qubits for Schrodingersation, $6$ qubits for approximating the original Fokker-Planck equation. 
Notably, we will numerically demonstrate below that with a noiseless Hamiltonian simulation, we can simulate a 1D time-dependent PDE only \emph{with a total of $20$ qubits}. Such an amount of noisy qubits is already available at present. 
However, to actually implemented it on a quantum hardware, we need e.g., error correcting techniques to deal with the still noisy environment for Hamiltonian simulations, which will be left as future investigations.
In Figure~\ref{fig::fp_eg_3}, we visualise the observables as well as their errors (compared with theoretically exact values) for the case (3) in which $g$ is a cubic function in time and $\beta$ is linear in time. With a sufficiently small $\omega$ in the initial state preparation, the error is around $10^{-2}$. The simulation results for two simpler cases (1) and (2) are deferred to Appendix~\ref{app::numerics}, in which the error has an order $10^{-3}$.

\begin{figure}
\begin{subfigure}[b]{0.65\textwidth}
\includegraphics[width=\textwidth]{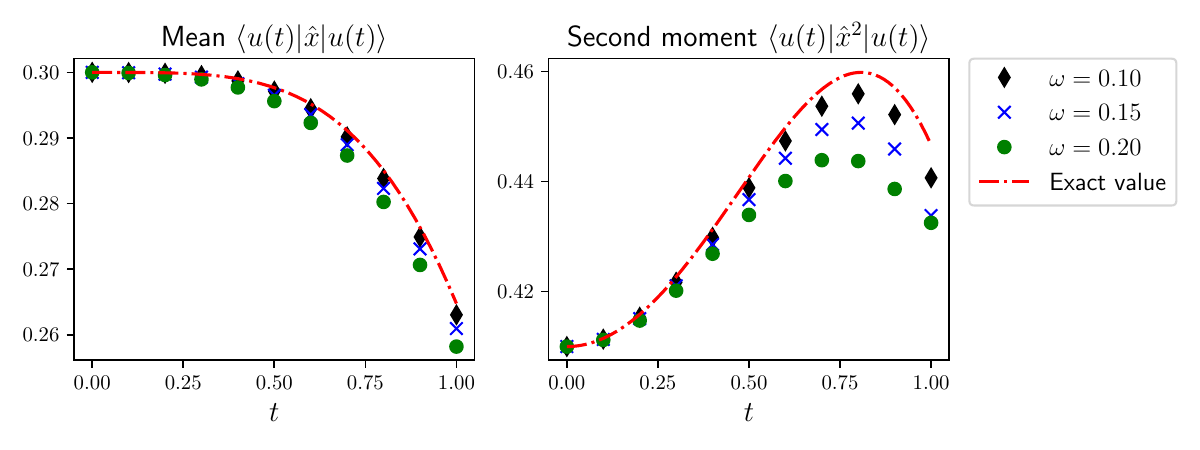}
\caption{Observables with respect to time}
\end{subfigure}
~
\begin{subfigure}[b]{0.65\textwidth}
\includegraphics[width=\textwidth]{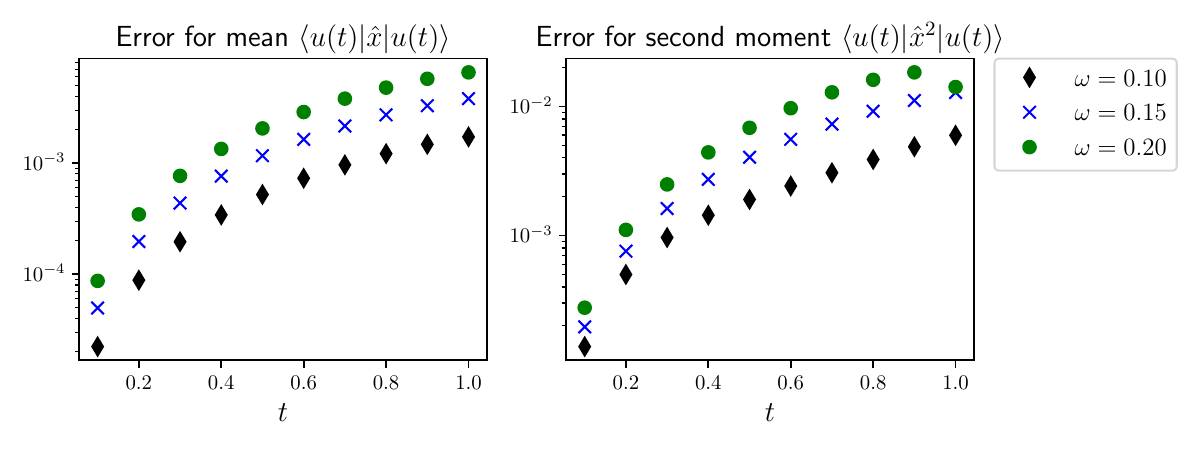}
\caption{Errors for observables}
\end{subfigure}
\caption{We visualise observables and their corresponding errors for quantum algorithms with various imperfect clock modes ($s$-mode). This picture corresponds to the third case in Eq.~\eqref{eqn::fp_eg}, namely, $g(t) = 0.5 s^3, \beta(t) = 0.3 s$.}
\label{fig::fp_eg_3}
\end{figure}

\section{Discussion}
\label{section::discussion}

We have shown how to turn any linear non-autonomous unitary system into an autonomous system obeying unitary dynamics, in one higher dimension. Combined with our method for turning any linear dynamical system into one obeying unitary dynamics in one higher dimension (Schr\"odingerisation), this means that the evolution of any linear system with time-dependent coefficients can be simulated on a quantum simulator with a single time-independent Hamiltonian with at most two additional modes for homogeneous systems with a first-order time derivative. For inhomogeneous systems, there is an additional single qubit, and one additional qubit for each higher-order time derivative. This means that any methods for the simulation of time-independent Hamiltonians can be immediately adapted to this time-dependent setting. \\

Apart from applications to ODEs and PDEs, there are many other interesting applications to explore in future work. These include further applications to both closed and open quantum systems, adiabatic computation, optimisation, and time-dependent iterative methods, that can also be applied to stochastic problems and applications of random matrices and many more.

\section*{Acknowledgements}
NL thanks Lorenzo Maccone for insightful discussions. YC is sponsored by the Shanghai Pujiang Program.
SJ was partially supported by the NSFC grant No.12031013, the Shanghai Municipal Science
and Technology Major Project (2021SHZDZX0102), and the Innovation Program of Shanghai Municipal Education Commission (No. 2021-01-07-00-02-E00087). NL acknowledges funding from the Science and Technology Program of Shanghai, China (21JC1402900).
The authors are also supported by the Shanghai Jiao Tong University 2030 Initiative and the Fundamental Research Funds for the Central
Universities.
\bibliography{Review}

\appendix

\section{Alternative protocol for the special case of the Hamiltonian that commutes at different times}
In the most general case, $[\vect{H}_{\tau}, \vect{H}_{\tau'}]\neq 0$. We can also have a simplified scenario where we have the commutativity between the Hamiltonian at different times $[\vect{H}_{\tau}, \vect{H}_{\tau'}]= 0$, which can be satisfied for $\vect{H}(t)=g(t)\vect{h}$ where $\vect{h}$ is a time-independent Hamiltonian. Then the corresponding unitary operation is 
\begin{align*}
    U(t)=e^{-i\int_0^t g(\tau) d \tau \vect{h}}=e^{-iG(t) t \vect{h}}, \qquad G(t) \equiv \frac{1}{t}\int_0^t g(\tau) d \tau. 
\end{align*}
where no time-ordering operation is required since the Hamiltonian commutes at different times. 
Our aim is to simulate the time evolution
\begin{align*}
  U(t)|\psi(0)\rangle=|\psi(t)\rangle. 
\end{align*}
Now let us define two new time-independent Hamiltonians that act on the original system plus an ancillary mode 
\begin{align*}
    &\hat{H}_1=\mathbf{1} \otimes \hat{p}_s, \qquad U_1(t)=e^{-i\hat{H}_1t} \\
    &\hat{H}_2=\vect{h} \otimes G(\hat{s}), \qquad U_2(t)=e^{-i\hat{H}_2t}, \qquad [\hat{x}, \hat{p}_x]=i\mathbf{1},
\end{align*}
where $\hat{s}$ is a quadrature operator acting on the ancillary mode. We now define an initial state
\begin{align*}
    |0_s\rangle |\psi(0)\rangle
\end{align*}
where $|0_s\rangle$ is the $s=0$ eigenstate so $\hat{s}|0_s\rangle=0$. \\

Applying the unitary $U_2(t)U_1(t)$ to this initial state
\begin{align*}
    U_2(t)U_1(t)|0_s\rangle |\psi(0)\rangle=U_2(t)|t_s\rangle|\psi(0)\rangle=e^{-iG(t)t \vect{h} \otimes \mathbf{1}_s}|t_s\rangle|\psi(0)\rangle.
\end{align*}
Tracing out the first (ancillary) register, we end up with the state
\begin{align*}
    e^{-iG(t)t\vect{h}}|\psi(0)\rangle=U(t)|\psi(0)\rangle,
\end{align*}
which is the initial state that we desired. \\

This means that with the simulation of the unitary system of $D$ qumodes with a time-dependent Hamiltonian, then we need only $D+1$ qumodes and access to the two unitary operations. The most difficult part is to engineer the Hamiltonian $\hat{H}_2$ when $G(t)$ is a complicated function of $t$. We note that this is a different method to the method considered in the main text. This method is more closely akin to solving the first equation first, then inserting the constant value into the second equation in
\begin{align*}
    & \frac{ds}{dt}=1, \quad s_0=t_0 \\
    & \frac{d \vect{u}}{dt}=\vect{H}(s) \vect{u}.
\end{align*}
By solving the first equation first, one gets $s=T$ for some final time $T$. This is equivalent to the $U_1(T)$ step. Then this means $H(s) \rightarrow H(T)$ becomes a time-independent Hamiltonian (since $T$ is now a constant). This method is not equivalent to that considered in the main text since in general $[\mathbf{1}\otimes \hat{p}_s, \vect{H}(\hat{s})] \neq 0$.

\section{Additional proofs}

\subsection{Proof of Lemma \ref{lemma::error_classical}}
\label{appendix::proof_classical}

Let us define $V(t,s):=\mathcal{U}_{s,s-t}\vect{y}_0$. 
Clearly $V(t,t)=\mathcal{U}_{t,0}\vect{y}_0=\vect{y}(t)$.  Using Taylor expansion on $V(t,s)$ at $s=t$,
\begin{align*}
\vect{y}_{\omega}(t) =& \int^{\infty}_{-\infty}\delta_{\omega}(s-t)V(t,s) ds\\
=& \int^{\infty}_{-\infty}\delta_{\omega}(s-t)
[V(t,t)+\partial_sV(t,t) (s-t)+\frac{1}{2}\partial_{ss}V(t,t)(s-t)^2+o((s-t)^2)]\,ds
\\
=& V(t,t)+ \mu \partial_s V(t,t) + \frac{\omega^2}{2}\partial_{ss}V(t,t)+o(\omega^2),
\end{align*}
using $\mu=\int (s-t)\delta_{\omega}(s-t)ds$. Now, for good state preparation, $\mu \sim 0$, but we can keep it more general since we see later the dominant contribution to $\vect{y}_{\omega}$ due to $\mu$ drops out when computing the quantum fidelity in any case. \\

We can rewrite the terms $\partial_{s} V(t,t)$ and $\partial_{ss} V(t,t)$ in the following way. We can first decompose $\mathcal{U}_{s, s-t}=\mathcal{U}_{s,0}\mathcal{U}_{0, s-t}$ and using Eq.~\eqref{eq:vequation} for $\mathcal{U}_{s,0}$ and $\mathcal{U}_{0,s-t}$
Using Eqs.~\eqref{eq:vequation} and \eqref{eq:vequation2}:
\begin{align*}
\partial_{s} V(t,s) &= -i \vect{H}(s) \mathcal{U}_{s, s-t} \vect{y}_0 + \mathcal{U}_{s,s-t} \big(i \vect{H}({s-t})\big) \vect{y}_0, \\
\partial_{ss} V(t,s) &= -i \dot{\vect{H}}(s) \mathcal{U}_{s, s-t} \vect{y}_0 - \vect{H}^2(s) \mathcal{U}_{s, s-t} \vect{y}_0 \\
&\qquad + 2\vect{H}(s) \mathcal{U}_{s,s-t} \vect{H}(s-t) \vect{y}_0 - \mathcal{U}_{s,s-t} \vect{H}^2(s-t) \vect{y}_0 + i \mathcal{U}_{s,s-t} \dot{\vect{H}}(s-t) \vect{y}_0.
\end{align*}
By choosing $s=t$, 
\begin{align}
\label{eqn::deri_V}
\begin{aligned}
\partial_{s} V(t,t) =& -i \vect{H}(t) \vect{y}(t) + i\mathcal{U}_{t,0} \vect{H}(0) \vect{y}_0 \\
\partial_{ss} V(t,t) =& -i \dot{\vect{H}}(t) \mathcal{U}_{t, 0} \vect{y}_0 - \vect{H}^2(t) \mathcal{U}_{t, 0} \vect{y}_0 + 2 \vect{H}(t) \mathcal{U}_{t,0} \vect{H}(0) \vect{y}_0 - \mathcal{U}_{t,0} \vect{H}^2(0) \vect{y}_0 + i \mathcal{U}_{t,0} \dot{\vect{H}}(0) \vect{y}_0\\
=&-i \dot{\vect{H}}(t) \vect{y}_t - \vect{H}^2(t) \vect{y}_t + 2 \vect{H}(t) \mathcal{U}_{t,0} \vect{H}(0) \vect{y}_0 - \mathcal{U}_{t,0} \vect{H}^2(0) \vect{y}_0 + i \mathcal{U}_{t,0} \dot{\vect{H}}(0) \vect{y}_0.
\end{aligned}
\end{align}
First note that since both $\vect{y}$ and $\vect{y}_{\omega}$ evolve under a unitary operation, their $l_2$ norms are preserved under the evolution, so $\|\vect{y}(t)\|=\|\vect{y}_0\|=1$, where we can assume $\|\vect{y}_0\|=1$ without losing generality, and $\|\vect{y}_{\omega}(t)\|=\|\vect{y}_{\omega}(0)\|$. By definition, $\vect{y}_{\omega}(0)=\int \delta_{\omega}(s)ds \vect{y}_0=\vect{y}_0$ when we choose the definition of $\delta_{\omega}$ satisfying Eq.~\eqref{eq:deltaomega2}. Thus the overlap between the ideal state $|y(t)\rangle=\vect{y}(t)$ and the approximate state $|y_{\omega}(t)\rangle=\vect{y}_{\omega}(t)$ is 
\begin{align*}
 &\ \langle y(t) |y_{\omega}(t)\rangle \\
=&\ 1+i \mu \big(\expval{\vect{H}(0)} - \expval{\vect{H}(t)}\big)\\
&\qquad + \frac{\omega^2}{2} \bigg(-i \expval{\dot{\vect{H}}(t)} - \expval{\vect{H}^2(t)} + 2 \bra{y(t)} \vect{H}(t) \mathcal{U}_{t,0} \vect{H}(0) \ket{y(0)} - \expval{\vect{H}^2(0)} + i \expval{\dot{\vect{H}}(0)}\bigg) + \order{\omega^3} \\
=&\ 1 - \frac{\omega^2}{2} \mathsf{C}_{\text{R}} + i \mathsf{C}_{\text{Im}} + \order{\omega^3},
\end{align*}
where we had used the fact that $\ket{y(t)}$ is a normalised state in the second line, and we define $\expval{\vect{H}^2(t)} := \ev{\vect{H}^2(t)}{y(t)}$, $\expval{\dot{\vect{H}}(t)} := \ev{\dot{\vect{H}}(t)}{y(t)}$ for any $t\in \mathbb{R}$, and we had introduced short-hand notations
\begin{align*}
\mathsf{C}_{\text{R}} &=\ \expval{\vect{H}^2(t)}+\expval{\vect{H}^2(0)}-2 \bra{y(t)} \vect{H}(t) \mathcal{U}_{t,0} \vect{H}(0) \ket{y(0)},\\
\mathsf{C}_{\text{Im}} &=\mu \big(\expval{\vect{H}(0)} - \expval{\vect{H}(t)}\big) +\frac{\omega^2}{2} \Big(-\expval{\dot{\vect{H}(t)}} + \expval{\dot{\vect{H}(0)}}\Big) = o(\omega).\  
\end{align*}
By direct computation,
\begin{align*}
\text{Fid}\big(|y_{\omega}(t)\rangle, |y(t)\rangle\big) &= \abs{\langle y(t) |y_{\omega}(t)\rangle}^2 \\
&= \big(1-\frac{\omega^2}{2} \mathsf{C}_{\text{R}}\big)^2 - \mathsf{C}_{\text{Im}}^2 + \order{\omega^3} \\
&= 1 - \omega^2 \mathsf{C}_{\text{R}} + o(\omega^2),
\end{align*}
where we had used the fact that the bias $\mu = o(\omega)$. Here $\mathsf{C}_{\text{R}}$ depends on the variance of the initial and final Hamiltonians, which of course are all bounded for realistic physical systems. For physical systems, the norms of the initial and final Hamiltonians are also bounded, and the last term $\mathsf{C}_{\text{R}}$ is also bounded. 

\subsection{Proof of Lemma~\ref{lemma::error_quantum}}
\label{appendix::proof_quantum}

The calculation is similar to Lemma \ref{lemma::error_classical}.
Recall that $\gamma_{\omega}(t) =
 \int \delta_{\omega}(t-s)\ \mathcal{U}_{s,s-t} \dyad{\vect{y}_0} \mathcal{U}_{s,s-t}^\dagger\ d s,$ and we want to quantify the overlap between $\gamma_{\omega}(t)$ and $\ket{y(t)} := \mathcal{U}_{t,0}\ket{y_0}$. After introducing 
\begin{align*}
\widetilde{V}(t,s) := \mathcal{U}_{s,s-t} \dyad{\vect{y}_0} \mathcal{U}_{s,s-t}^\dagger,
\end{align*}
we can easily derive the following 
\begin{align*}
\gamma_{\omega}(t) = \dyad{y(t)} + \mu \partial_s\widetilde{V}(t,t) + \frac{\omega^2}{2} \partial_{ss} \widetilde{V}(t,t) + o(\omega^2).
\end{align*}
 
Then we can similarly calculate $\partial_s\widetilde{V}(t,t)$ and $\partial_{ss}\widetilde{V}(t,t)$ similarly:
 \begin{align*}
 \partial_s\widetilde{V}(t,t) &= \partial_s V(t,t) V(t,t)^\dagger + h.c. \\
 &\myeq{\eqref{eqn::deri_V}}  -i \vect{H}(t) \dyad{y(t)} + i\mathcal{U}_{t,0} \vect{H}(0) \ketbra{y_0}{y(t)} + h.c.
\end{align*}
and
\begin{align*}
 & \partial_{ss}\widetilde{V}(t,t) = \partial_{ss} V(t,t) V(t,t)^\dagger + \partial_{s} V(t,t) \partial_s V(t,t)^\dagger + h.c.\\
 &\myeq{\eqref{eqn::deri_V}} -i \dot{\vect{H}}(t) \dyad{y(t)} - \vect{H}^2(t) \dyad{y(t)} + 2 \vect{H}(t) \mathcal{U}_{t,0} \vect{H}(0) \ketbra{y_0}{y(t)} \\
 &\qquad - \mathcal{U}_{t,0} \vect{H}^2(0) \ketbra{y_0}{y(t)} + i \mathcal{U}_{t,0} \dot{\vect{H}}(0) \ketbra{y_0}{y(t)} \\
 &\qquad  +\Big(-i \vect{H}(t) |y(t)\rangle + i\mathcal{U}_{t,0} H(0) |y_0\rangle\Big) \Big(i \bra{y(t)} \vect{H}(t)  - i \bra{y_0} \vect{H}(0) \mathcal{U}_{t,0}^\dagger\Big)  + h.c.
\end{align*}
As a remark, \enquote{+h.c.} means adding the Hermitian conjugate of all previous terms in the equation.

Then the Fidelity between the reduced quantum state $\gamma_{\omega}(t)$ and $\ket{y(t)}$ is 
 \begin{align*}
& \expval{\gamma_{\omega}(t)}{y(t)} \\
=& 1 + \mu \Big(-i \expval{\vect{H}(t)}{y(t)}  + i \expval{\vect{H}(0)}{y_0} + h.c.\Big) \\
& + \frac{\omega^2}{2} \left(\begin{aligned}
&-i \expval{\dot{\vect{H}}(t)}{y(t)} - \expval{\vect{H}^2(t)}{y(t)} \\
& + 2 \bra{y(t)} \vect{H}(t) \mathcal{U}_{t,0} \vect{H}(0) \ket{y_0} - \expval{\vect{H}^2(0)}{y_0} + i  \expval{\dot{\vect{H}}(0)}{y_0} \\
 &  +\Big(-i \expval{\vect{H}(t)}{y(t)} + i \expval{\vect{H}(0)}{y_0}\Big) \Big(i \expval{\vect{H}(t)}{{y(t)}}  - i \expval{\vect{H}(0)}{{y_0}}\Big)  + h.c.
\end{aligned}\right) + o(\omega^2)\\
=& 1 + \omega^2 \left(\begin{aligned}
&  - \expval{\vect{H}^2(t)}{y(t)} \\
& + 2\text{Re}\bra{y(t)} \vect{H}(t) \mathcal{U}_{t,0} \vect{H}(0) \ket{y_0} - \expval{\vect{H}^2(0)}{y_0}  \\
 &  + \Big(\expval{\vect{H}(t)}{y(t)} - \expval{\vect{H}(0)}{y_0}\Big)^2\Big) 
\end{aligned}\right) + o(\omega^2).
 \end{align*}
This proves our lemma by substituting the above expressions via the constant $\mathsf{C}$.

\section{Page-Wotters mechanism} \label{app:PW}

Although at first our protocol might seem similar to the Page-Wotters mechanism,  it is in fact a different protocol, but which still recovers the state $|y(t)\rangle$. The Page-Wooter's mechanism relies on the preparation of the \textit{ground state} of $\vect{\bar{H}}=\mathbf{1}\otimes \hat{p}_s+\vect{H}(\hat{s})$. Here we discretise in $s$ and can write the ground state as 
\begin{align*}
    \vect{\bar{H}}|\Psi\rangle=0, \qquad |\Psi\rangle=\frac{1}{\sqrt{N}}\sum_{i=-N/2}^{N/2}|y(s_i)\rangle \otimes |i\rangle,
\end{align*}
which is a quantum superposition of the time-dependent state $|y(t)\rangle$ entangled with the clock register. To recover $|y(t)\rangle$, one would require a method to prepare the ground state of a given $\vect{\bar{H}}$ (which can even be NP-hard depending on the Hamiltonian $\vect{H}(t)$). Given $|\Psi\rangle$, one would further need a projective measurement onto $|s=t\rangle$ so 
\begin{align*}
    |y(t)\rangle \propto \langle s=t|\Psi\rangle.
\end{align*}
This is clearly very resource intensive compared to our protocol. Its objective is also different to ours, namely to create a time-independent state $|\Psi\rangle$, whereas we only aim to create a state $|y(t)\rangle$ that evolves according to a time-independent Hamiltonian. 

\section{Simulation details and additional pictures}
\label{app::numerics}
We use the classical computer to emulate the above quantum algorithm in Theorem~\ref{thm:three} (see also Figure~\ref{fig:perfect}). To approximate the clock mode ($s$-mode), we use the Galerkin projection method to map the target PDE via a family of fixed discrete basis functions. For simplicity, we use the following orthonormal basis on $L^2(\mathbb{R})$,
\begin{align*}
\phi_n(x) = \frac{H_n(\frac{x}{\varsigma}) \exp(-\frac{x^2}{2\varsigma^2})}{C_n},\qquad C_n = (\pi\varsigma^2)^{1/4} \sqrt{n!} 2^{n/2},
\end{align*}
$H_n$ are Hermite polynomials, and $\varsigma$ is a tunable parameter. 
For the clock mode, the number of basis functions is denoted as $N_s$ and the parameter is $\varsigma_s$. Similar notations are being used for other continuous modes.
Similarly, the projection operator $\hat{P}_{>0}$ is also implemented via approximating this operator on the basis $\{\phi_n\}_{n=0}^{N_\eta -1}$.
Such a Galerkin projection basis method will enable us to simulate an one-dimensional PDE using quantum algorithms in Theorem~\ref{thm:three}. 
For large-scale PDE simulations using our quantum algorithms, a machine learning based simulator \cite{e_deep_2018} becomes necessary and will be left as future development.

\bigskip
{\noindent \bf Details about the 1D Fokker-Planck.}
\smallskip
Below are some formulas for the above 1D Fokker-Planck equations:
\begin{itemize}[leftmargin=1em]
\item (Observables). If $\vect{q} = \mathcal{N}(\mu, \sigma^2)$ is a Gaussian distribution, then the expectations of the corresponding normalised state are  
\begin{align*}
\expval{\hat{x}}{q} &=  \mu,\qquad \expval{\hat{x}^2}{q} = \frac{\sigma^2}{2} + \mu^2,
\end{align*}
where $\ket{q} = \vect{q}/\norm{\vect{q}}$.

\item (Evolution of the second moment). The second moment $M(t) := \int dx\  x^2 \vect{q}(t,x)$ of the 1D Fokker-Planck equation has the following explicit solution:
\begin{align*}
M(t) = e^{-2 G(t)} M(0) + e^{-2 G(t)} \int_{0}^{t} e^{2G(r)} 2\beta(r)\ dr
\end{align*}
and the mean $\mu(t) = \int dx\ x \vect{q}(t,x)$ has the following explicit solution:
\begin{align*}
\mu(t) = e^{-G(t)} \mu(0).
\end{align*}
Therefore, we use these exact values as reference to examine the accuracy of quantum algorithms.
\end{itemize}

We include simulation results for two simpler cases of the time functions $g, \beta$ below in Figures~\ref{fig::fp_1} and \ref{fig::fp_2}.
In these pictures, we only measure observables at a selective number of discrete time points, and compared them with the theoretically exact values, whose formulas are given above.

\begin{figure}[ht]
\begin{subfigure}[b]{0.65\textwidth}
\includegraphics[width=\textwidth]{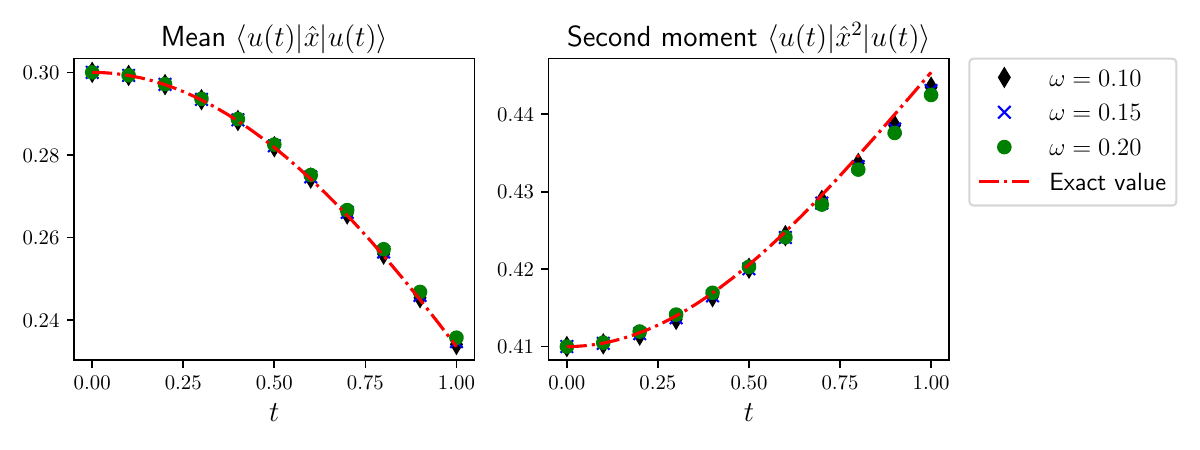}
\caption{Observables with respect to time}
\end{subfigure}
~
\begin{subfigure}[b]{0.65\textwidth}
\includegraphics[width=\textwidth]{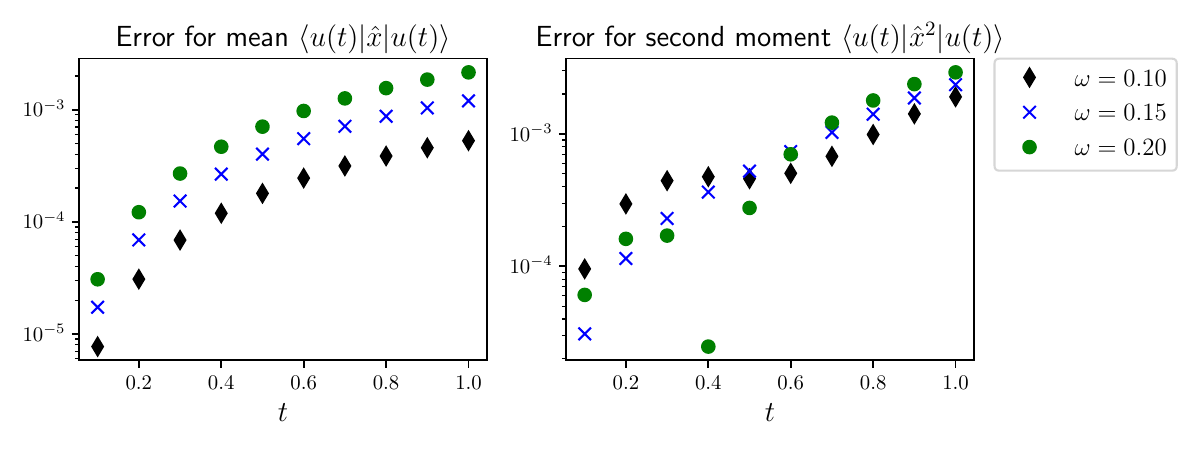}
\caption{Errors for observables}
\end{subfigure}
\caption{We visualise observables and their corresponding errors for quantum algorithms with various imperfect clock modes ($s$-mode). This picture corresponds to the first case in Eq.~\eqref{eqn::fp_eg}, namely, $g(t) = 0.5 s,\ \beta(t) = 0.5 s$.}
\label{fig::fp_1}
\end{figure}

\begin{figure}[ht]
\begin{subfigure}[b]{0.65\textwidth}
\includegraphics[width=\textwidth]{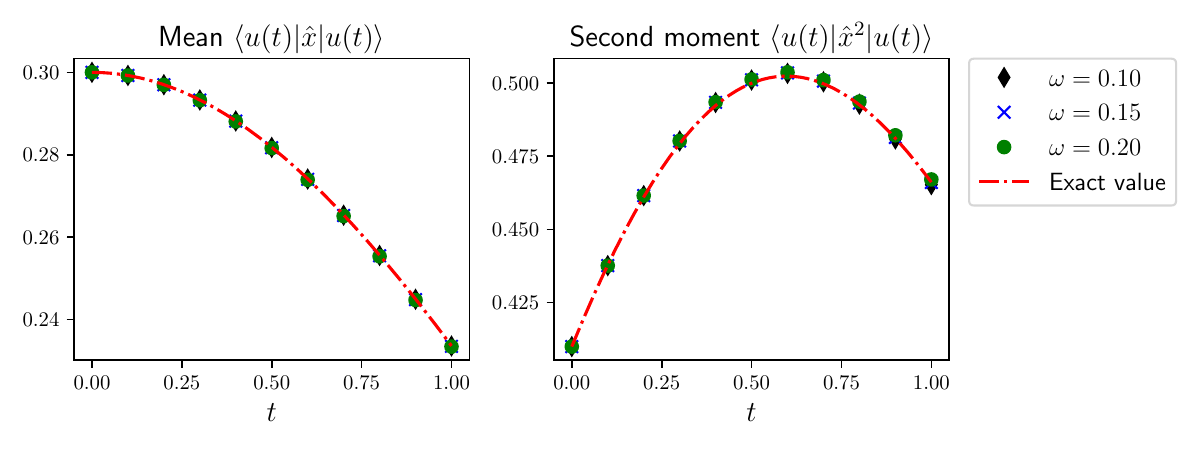}
\caption{Observables with respect to time}
\end{subfigure}
~
\begin{subfigure}[b]{0.65\textwidth}
\includegraphics[width=\textwidth]{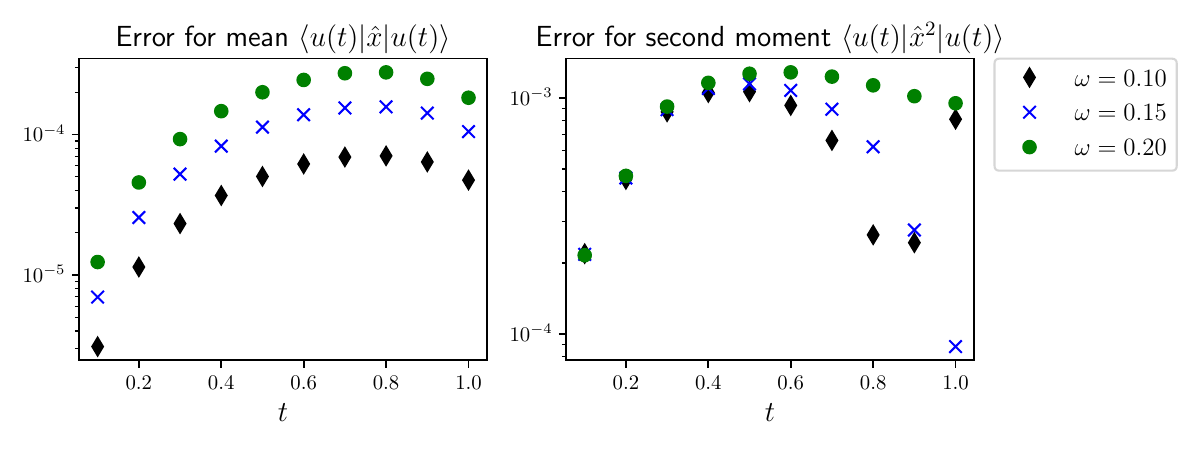}
\caption{Errors for observables}
\end{subfigure}
\caption{We visualise observables and their corresponding errors for quantum algorithms with various imperfect clock modes ($s$-mode). This picture corresponds to the second case in Eq.~\eqref{eqn::fp_eg}, namely, $g(t) = 0.5 s,\ \beta(t) = 0.3$.}
\label{fig::fp_2}
\end{figure}

\end{document}